%% file: EIG-Grad_article.tex
\documentclass[onefignum,onetabnum]{siamart171218}

\input{EIG-Grad_shared}

\begin{document}

\maketitle

\begin{abstract}
In this paper we propose an efficient stochastic optimization algorithm to search for Bayesian experimental designs such that the expected information gain is maximized. The gradient of the expected information gain with respect to experimental design parameters is given by a nested expectation, for which the standard Monte Carlo method using a fixed number of inner samples yields a biased estimator. In this paper, applying the idea of randomized multilevel Monte Carlo (MLMC) methods, we introduce an unbiased Monte Carlo estimator for the gradient of the expected information gain with finite expected squared $\ell_2$-norm and finite expected computational cost per sample. Our unbiased estimator can be combined well with stochastic gradient descent algorithms, which results in our proposal of an optimization algorithm to search for an optimal Bayesian experimental design. Numerical experiments confirm that our proposed algorithm works well not only for a simple test problem but also for a more realistic pharmacokinetic problem.
\end{abstract}



\section{Introduction}\label{sec:intro}
In this paper we study optimization of Bayesian experimental designs which aim to maximize the expected amount of information experimental outcomes convey about unobservable, or hidden/latent, random variables of interest by carefully designing an experimental setup. Here we measure the expected amount of information by the Shannon's expected information gain whose definition is given below. Our motivation comes from applications to a number of disciplines, such as mechanical engineering \cite{R03}, neuroscience \cite{VTHR12}, bioinformatics \cite{SPPP13}, psychology \cite{MCP13}, and pharmacokinetics \cite{RDTP14,RDP15} among many others.

Let $\theta=(\theta_1,\ldots,\theta_s)\in \Theta\subseteq \RR^s$ be a vector of continuous unobservable random variables, and we denote the prior probability density of $\theta$ by $\pi_0(\theta)$. The information entropy, or the differential entropy, of $\theta$ is defined by
\[ \EE_\theta\left[-\log \pi_0(\theta)\right] = \int_{\Theta}-\pi_0(\theta)\log \pi_0(\theta) \rd \theta. \]
Let us consider a situation where, by conducting some experiments under an experimental design $\xi$, an observation $Y=(Y_1,\ldots,Y_t)\in \Ycal\subseteq \RR^t$ is obtained according to the forward model
\begin{align}\label{eq:forward} 
Y = f_\xi(\theta,\epsilon), 
\end{align}
where $\epsilon=(\epsilon_1,\ldots,\epsilon_{s'})\in \Ecal \subseteq \RR^{s'}$,  representing the observation noise, is another vector of continuous random variables with its density $\varphi(\epsilon)$, and $f_\xi$ is a deterministic bi-variate function parametrized by the design $\xi$, possibly with multiple outputs. Here we assume that the experimental design $\xi$ is controllable and can be chosen as an element in an open set $\Xcal\subset \RR^{d}$. Throughout this paper, we assume that the domain $\Ycal$ is independent of $\xi$, that $\epsilon$ is independent both of $\theta$ and $\xi$, and also that the likelihood function $\rho(Y\mid \theta,\xi)$ is strictly positive and can be computed explicitly with unit cost for any pair of $\theta, \xi$ and $Y$. As is well known, Bayes' theorem states that the posterior probability density of $\theta$ given $Y$, denoted by $\pi^{Y\mid \xi}$, is given by
\begin{align}\label{eq:bayes} \pi^{Y\mid \xi}(\theta) = \frac{\rho(Y\mid \theta,\xi)\pi_0(\theta)}{\rho(Y\mid \xi)}, \end{align}
with $\rho(Y\mid \xi)$ being the marginal likelihood of $Y$, i.e.,
\[ \rho(Y\mid \xi)=\EE_{\theta}\left[\rho(Y\mid \theta,\xi)\right]=\int_{\Theta}\rho(Y\mid \theta,\xi)\pi_0(\theta)\rd \theta,\]
see for instance \cite{S10}. Then, the posterior information entropy of $\theta$ after observing $Y$ is given by
\[ \EE_{\theta\mid Y,\xi}\left[-\log \pi^{Y\mid \xi}(\theta)\right] = \int_{\Theta}-\pi^{Y\mid \xi}(\theta) \log \pi^{Y\mid \xi}(\theta) \rd \theta, \]
and hence, the expected posterior information entropy of $\theta$ by conducting an experiment under an experimental design $\xi$ is given by integrating the posterior information entropy of $\theta$ over $Y$ using the marginal likelihood $\rho(Y\mid \xi)$, i.e., 
\[ \EE_{Y\mid \xi}\EE_{\theta\mid Y,\xi}\left[-\log \pi^{Y\mid \xi}(\theta)\right] = \int_{\Ycal}\int_{\Theta}- \pi^{Y\mid \xi}(\theta) \log \pi^{Y\mid \xi}(\theta)\rd \theta\, \rho(Y\mid \xi) \rd Y. \]
Now the difference
\[ U(\xi):=\EE_\theta\left[-\log \pi_0(\theta)\right]- \EE_{Y\mid \xi}\EE_{\theta\mid Y,\xi}\left[-\log \pi^{Y\mid \xi}(\theta)\right] \]
is called the \emph{expected information gain}, the quantity originally introduced in \cite{L56} as a measure of experimental designs. By using Bayes' theorem \eqref{eq:bayes}, we see that $U(\xi)$ is equivalently given by
\begin{align}
U(\xi) & = \EE_\theta \EE_{Y\mid \theta, \xi}\left[ \log \rho(Y\mid \theta,\xi)\right]- \EE_{Y\mid \xi}\left[ \log \rho(Y\mid \xi)\right] \notag \\
& = \EE_\theta \EE_{Y\mid \theta,\xi}\left[ \log \rho(Y\mid \theta,\xi)\right]- \EE_{Y\mid \xi}\left[ \log \EE_{\theta}\left[ \rho(Y\mid \theta,\xi)\right] \right]. \label{eq:eig} 
\end{align}

The aim of Bayesian experimental designs is to construct an optimal experimental design $\xi=\xi^*$ which maximizes the expected information gain $U$ \cite{CV95}. As can be seen from the second term of \eqref{eq:eig}, however, estimating $U(\xi)$ is inherently a nested expectation problem with an outer expectation with respect to $Y$ and an inner expectation with respect to $\theta$, which has been considered computationally challenging. The standard, nested Monte Carlo method generates $N$ outer random samples for $Y$ first and then, for each sample of $Y$, generates $M$ inner random samples for $\theta$. To estimate $U(\xi)$ with root-mean-square accuracy $\varepsilon$,\footnote{Here and in what follows, the difference between the noise $\epsilon$ and the accuracy $\varepsilon$ should not be confused.} we typically need $N=O(\varepsilon^{-2})$ and $M=O(\varepsilon^{-1})$, resulting in a total computational complexity of $O(\varepsilon^{-3})$ \cite{R03,BDELT18,RCYWW18}. Recently there have been some attempts in \cite{GHI20,BDETxx} to reduce this cost to $O(\varepsilon^{-2})$ or $O(\varepsilon^{-2}(\log \varepsilon^{-1})^2)$ by applying a multilevel Monte Carlo (MLMC) method \cite{G08,G15} in conjunction with Laplace approximation-based importance sampling \cite{LSTW13}. Here the difference between the orders of complexity for the MLMC method is a direct consequence from the basic MLMC theorem, see for instance \cite[Theorem~2.1]{G15}, which itself depends on the properties of the constructed MLMC estimators. Nevertheless, these results are an intermediate step towards an efficient construction of optimal experimental designs since design optimization has been left behind.

In this paper we deal with this optimization problem more directly. More precisely, under the assumption that the experimental setup, or the set of design parameters, $\xi$ lives in a continuous space such that $U$ is differentiable with respect to $\xi$, we consider applying stochastic gradient descent optimizations to search for an optimal $\xi$. As we shall see, the gradient $\nabla_\xi U$ is again given by a nested expectation, for which the standard, nested Monte Carlo method using a fixed number of inner samples yields a biased estimator. By applying an unbiased MLMC method from \cite{RG15}, a randomized version of the original MLMC method, we can construct an unbiased estimator of $\nabla_\xi U$. This way, in this paper, we arrive at a stochastic gradient-based optimization algorithm in which unbiased random samples to estimate $\nabla_\xi U$ are generated at each iteration step.

Here we have to mention that the idea of using stochastic gradient-based methods in Bayesian experimental designs already exists in the literature \cite{HM14,FJBHTRG19,FJOTR20,CDELT20,KG20}. In particular, a work by Carlon et al.\ \cite{CDELT20} takes a similar standpoint in that an analytical expression of the gradient $\nabla_\xi U$ is derived and then stochastic gradient-based method is applied in conjunction with Monte Carlo estimation of $\nabla_\xi U$. However, the expression of $\nabla_\xi U$ given in \cite[Proposition~1]{CDELT20} is proven only for the additive Gaussian noise $\epsilon$, that is, the case where the forward model is given by the form $Y=f_{\xi}(\theta)+\epsilon$ with $\epsilon\sim N(0,\Sigma),$ and the standard (biased) Monte Carlo estimator is used at each iteration step within stochastic gradient-based methods. In this paper we consider a more general form of the forward model as shown in \eqref{eq:forward}, which is useful in some applications \cite{RDTP14,RDP15}. Moreover, given that stochastic gradient-based methods are usually established under the assumption that each sample is drawn from the underlying true distribution, using an unbiased estimator of $\nabla_\xi U$ should be favorable, and by doing so, we do not need to take care of the bias-variance tradeoff. Although application of MLMC methods to stochastic approximation algorithms have been investigated recently in \cite{F16,DG19}, neither of them considers using a randomized MLMC method to generate unbiased random samples at each iteration step.

The rest of this paper is organized as follows. In Section~\ref{sec:gradient}, we provide an analytical expression of the gradient $\nabla_\xi U$ and also briefly review some of stochastic gradient-based optimization methods. Although there are a number of stochastic optimization algorithms, one can use any of them in our proposal to optimize Bayesian experimental designs (Algorithm~\ref{alg:mlmc_sgd}), and we do not give any recommendation on which method should be used in our algorithm, since it is not the objective of this paper. Again we emphasize that the main contribution of this paper is to provide an unbiased estimator for the gradient $\nabla_\xi U$, which is non-trivial but a key assumption in stochastic gradient-based optimization. In Section~\ref{sec:mlmc}, after introducing a standard, nested Monte Carlo estimator of $\nabla_\xi U$, which is obviously biased, we provide an unbiased, multilevel Monte Carlo estimator of $\nabla_\xi U$ and prove under some conditions that our estimator has a finite expected squared $\ell_2$-norm with finite computational cost per sample. Our proposal for optimizing Bayesian experimental designs is given in Algorithm~\ref{alg:mlmc_sgd}. To demonstrate the effectiveness of our proposed algorithm, we conduct numerical experiments not only for a simple test problem but also for a more realistic pharmacokinetic (PK) problem in Section~\ref{sec:numerics}. We conclude this paper with some remarks in Section~\ref{sec:end}.

\section{Stochastic gradient-based optimization}\label{sec:gradient}
\subsection{Gradient of expected information gains}
In what follows, we give an explicit form of the gradient $\nabla_\xi U$. As a preparation, let us rewrite the expected information gain $U(\xi)$ according to \eqref{eq:forward} in the following way. First, by noting that generating $Y$ randomly conditional on $\theta$ and $\xi$ is equivalent to computing $f_\xi(\theta,\epsilon)$ for a randomly generated $\epsilon$ with both $\theta$ and $\epsilon$ given, the independence between $\epsilon$ and the pair $(\theta,\xi)$ ensures that the first term of \eqref{eq:eig} is equal to
\[ \EE_{\theta}\EE_{\epsilon}\left[ \log \rho(f_{\xi}(\theta,\epsilon)\mid \theta,\xi)\right]=\EE_{\theta,\epsilon}\left[ \log \rho(f_{\xi}(\theta,\epsilon)\mid \theta,\xi)\right]. \]
Similarly, generating $Y$ randomly conditional only on $\xi$ is equivalent to computing $f_\xi(\theta,\epsilon)$ for randomly generated $\theta$ and $\epsilon$ with a fixed $\xi$. Therefore, by denoting an i.i.d.\ copy of $\theta$ by $\theta'$, the second term of \eqref{eq:eig} is equal to
\[ \EE_{Y\mid \xi}\left[ \log \EE_{\theta'}\left[ \rho(Y\mid \theta',\xi)\right] \right]=\EE_{\theta,\epsilon}\left[ \log \EE_{\theta'}\left[ \rho(f_\xi(\theta,\epsilon)\mid \theta',\xi)\right] \right]. \]
Thus we end up with the following expression of $U(\xi)$:
\begin{align} \label{eq:eig2}
U(\xi)=\EE_{\theta, \epsilon}\left[ \log \rho(f_{\xi}(\theta,\epsilon)\mid \theta,\xi)\right]- \EE_{\theta, \epsilon}\left[ \log \EE_{\theta'}\left[ \rho(f_{\xi}(\theta,\epsilon)\mid \theta',\xi)\right] \right].
\end{align}
As we have stated in the previous section, we assume throughout this paper that the likelihood function can be computed explicitly with unit cost for any pair of inputs. Here we give some examples for which such an explicit computation of the likelihood function is possible.

\begin{example}[Additive noise]
Let us consider a forward model given by
\[ f_{\xi}(\theta,\epsilon) = g_{\xi}(\theta)+\epsilon, \]
for a uni-vatiate function $g_{\xi}: \Theta\to \Ycal\, (=\RR^t)$ and $\epsilon\sim N(0,\Sigma)$ with a covariance matrix $\Sigma$ and $s'=t$. Then, denoting the density of $\epsilon$ by $\varphi$, we have
\[ \rho(f_{\xi}(\theta,\epsilon)\mid \theta',\xi) = \varphi(\epsilon+g_{\xi}(\theta)-g_{\xi}(\theta')),\]
with a special case $\rho(f_{\xi}(\theta,\epsilon)\mid \theta,\xi) = \varphi(\epsilon)$.
\end{example}

\begin{example}[Multiplicative noise]
Let $s'=t=1$ for simplicity, and consider a forward model given by
\[ f_{\xi}(\theta,\epsilon) = g_{\xi}(\theta)\times (1+\epsilon), \]
with $g_{\xi}: \RR^s\to \RR_{>0}$ and $\epsilon\sim N(0,\sigma^2)$. Denoting the density of $\epsilon$ by $\varphi$, we have
\[ \rho(f_{\xi}(\theta,\epsilon)\mid \theta',\xi) = \varphi\left(\frac{g_{\xi}(\theta)}{g_{\xi}(\theta')}(1+\epsilon)-1\right),\]
with a special case $\rho(f_{\xi}(\theta,\epsilon)\mid \theta,\xi) = \varphi(\epsilon)$.
\end{example}

\begin{example}[Mixture of additive and multiplicative noises]\label{exm:mix}
Finally, for $t=1$ and $s'=2$, i.e., $\epsilon=(\epsilon_1,\epsilon_2)\in \RR^2$, let us consider a forward model described by
\[ f_{\xi}(\theta,\epsilon) = g_{\xi}(\theta)\times (1+\epsilon_1)+\epsilon_2, \]
with $g_{\xi}: \RR^s\to \RR_{>0}$, $\epsilon_1\sim N(0,\sigma_1^2)$ and $\epsilon_2\sim N(0,\sigma_2^2)$. Denoting the density of the standard normal random variable by $\varphi$, we have
\[ \rho(f_{\xi}(\theta,\epsilon)\mid \theta',\xi) = \varphi\left(\frac{g_{\xi}(\theta)\times (1+\epsilon_1)+\epsilon_2-g_{\xi}(\theta')}{\sqrt{|g_{\xi}(\theta')|^2\sigma_1^2+\sigma_2^2}}\right),\]
with a special case
\[\rho(f_{\xi}(\theta,\epsilon)\mid \theta,\xi) = \varphi\left(\frac{\epsilon_1 g_{\xi}(\theta)+\epsilon_2}{\sqrt{|g_{\xi}(\theta)|^2\sigma_1^2+\sigma_2^2}}\right). \]
\end{example}

Now we are ready to derive the gradient $\nabla_\xi U$. Note that our claim does not assume that the noise $\epsilon$ is additive and a Gaussian random variable, as discussed in the last two examples.

\begin{proposition}\label{prop:gradient}
Let $\theta'$ be an i.i.d.\ copy of $\theta$. Assume that the likelihood functions $\rho(f_{\xi}(\theta,\epsilon)\mid \theta,\xi)$ and  $\rho(f_{\xi}(\theta,\epsilon)\mid \theta',\xi)$ and their gradients $\nabla_{\xi}\rho(f_{\xi}(\theta,\epsilon)\mid \theta,\xi)$ and $\nabla_{\xi}\rho(f_{\xi}(\theta,\epsilon)\mid \theta',\xi)$ are all continuous with respect to $\theta, \theta',\epsilon$ and $\xi$. Then we have
\[ \nabla_\xi U(\xi) =\EE_{\theta, \epsilon}\left[\frac{\nabla_{\xi}\rho(f_{\xi}(\theta,\epsilon)\mid \theta,\xi)}{\rho(f_{\xi}(\theta,\epsilon)\mid \theta,\xi)}- \frac{\EE_{\theta'}\left[\nabla_{\xi} \rho(f_{\xi}(\theta,\epsilon)\mid \theta',\xi)\right]}{\EE_{\theta'}\left[ \rho(f_{\xi}(\theta,\epsilon)\mid \theta',\xi)\right]} \right]. \]
\end{proposition}
\begin{proof}
Under the continuity assumption on the likelihood function, the Leibniz integral rule applies and we have
\begin{align*}
\nabla_{\xi}U(\xi)& = \EE_{\theta, \epsilon}\left[ \nabla_{\xi} \log \rho(f_{\xi}(\theta,\epsilon)\mid \theta,\xi)\right]- \EE_{\theta, \epsilon}\left[ \nabla_{\xi}\log \EE_{\theta'}\left[ \rho(f_{\xi}(\theta,\epsilon)\mid \theta',\xi)\right] \right] \\
& = \EE_{\theta, \epsilon}\left[ \frac{\nabla_{\xi}\rho(f_{\xi}(\theta,\epsilon)\mid \theta,\xi)}{\rho(f_{\xi}(\theta,\epsilon)\mid \theta,\xi)}\right]- \EE_{\theta, \epsilon}\left[\frac{\nabla_{\xi}\EE_{\theta'}\left[ \rho(f_{\xi}(\theta,\epsilon)\mid \theta',\xi)\right]}{\EE_{\theta'}\left[ \rho(f_{\xi}(\theta,\epsilon)\mid \theta',\xi)\right]} \right] \\
& = \EE_{\theta, \epsilon}\left[ \frac{\nabla_{\xi}\rho(f_{\xi}(\theta,\epsilon)\mid \theta,\xi)}{\rho(f_{\xi}(\theta,\epsilon)\mid \theta,\xi)}\right]- \EE_{\theta, \epsilon}\left[\frac{\EE_{\theta'}\left[\nabla_{\xi} \rho(f_{\xi}(\theta,\epsilon)\mid \theta',\xi)\right]}{\EE_{\theta'}\left[ \rho(f_{\xi}(\theta,\epsilon)\mid \theta',\xi)\right]} \right].
\end{align*}
\end{proof}

As is clear from this proposition, because of the ratio of inner expectations, the gradient $\nabla_\xi U$ is inherently given by a nested expectation with an inner expectation with respect to $\theta'$ and an outer expectation with respect to $\theta$ and $\epsilon$. 

\subsection{Basics of stochastic gradient-based optimization}\label{subsec:gradient}
We recall that the aim of Bayesian experimental designs is to find an optimal experimental setup $\xi=\xi^*$ which satisfies
\[ \xi^* = \arg\max_{\xi\in \Xcal}U(\xi), \]
where we recall that an open set $\Xcal\subset \RR^{d}$ denotes the feasible domain of $\xi$. To achieve this goal, one of the reasonable approaches is to use some gradient-based optimization methods in which we set an initial experimental setup $\xi_0\in \Xcal$ and recursively update itself as
\[ \xi_{t+1}=g_t(\xi_t, \nabla_\xi U(\xi_t))\quad \text{for $t=0,1,\ldots$,}\]
until a certain stopping criterion is met. However, computing $\nabla_\xi U$ is already challenging since it is given by a nested expectation. As inferred from the results shown in the next section, it is possible to construct an antithetic MLMC estimator which efficiently estimates $\nabla_\xi U$, but we avoid such a ``pointwise'' accurate gradient estimation by using \emph{stochastic} gradient-based optimization methods. What we need here is an unbiased estimator of $\nabla_\xi U$ with finite variance and computational cost.

To simplify the presentation, let us define a vector of random variables
\begin{align}\label{eq:psi} \psi_\xi :=  \frac{\nabla_{\xi}\rho(f_{\xi}(\theta,\epsilon)\mid \theta,\xi)}{\rho(f_{\xi}(\theta,\epsilon)\mid \theta,\xi)}- \frac{\EE_{\theta'}\left[\nabla_{\xi} \rho(f_{\xi}(\theta,\epsilon)\mid \theta',\xi)\right]}{\EE_{\theta'}\left[ \rho(f_{\xi}(\theta,\epsilon)\mid \theta',\xi)\right]}, \end{align}
with $\theta\sim \pi_0$ and $\epsilon\sim \varphi$ being the underlying stochastic variables. It follows from Proposition~\ref{prop:gradient} that $\EE[\psi_\xi]=\nabla_\xi U(\xi)$. Suppose at this moment that we are able to generate i.i.d.\ random samples of $\psi_\xi$. We emphasize that random sampling of $\psi_\xi$ is far from trivial but we shall show in the next section that this is indeed possible. 

In stochastic gradient-based optimization methods, after setting an initial experimental setup $\xi_0\in \Xcal$, we recursively update itself as
\[ \xi_{t+1}=g_t(\xi_t, \psi_{\xi_t})\quad \text{for $t=0,1,\ldots$,}\]
or more generally,
\[ \xi_{t+1}=g_t\left(\xi_t, \frac{1}{N}\sum_{n=1}^{N}\psi_{\xi_t}^{(n)} \right)\quad \text{for $t=0,1,\ldots$,}\]
where $\psi_{\xi_t}^{(1)},\ldots,\psi_{\xi_t}^{(N)}$ are i.i.d.\ realizations of $\psi_{\xi_t}$ for a sample size $N\in \ZZ_{>0}$. This means that, at each iteration, we only need (rough) unbiased Monte Carlo estimate of $\EE[\psi_\xi]$ instead of the true value. There have been many examples for this recursion $g_t$ proposed in the literature.

For instance, one of the most classical methods due to Robbins and Monro \cite{RM51} is simply given by
\[ \xi_{t+1}=\Pi_{\Xcal}\left(\xi_t+a_t \cdot \frac{1}{N}\sum_{n=1}^{N}\psi_{\xi_t}^{(n)}\right) , \]
with a sequence of non-negative reals called \emph{learning rates} $a_0,a_1,\ldots$ such that
\[ \sum_{t=0}^{\infty}a_t=\infty\quad \text{and}\quad \sum_{t=0}^{\infty}a_t^2<\infty, \]
where $\Pi_{\Xcal}$ denotes the projection operator which maps the input to a closest point in $\Xcal$, i.e., $\Pi_{\Xcal}(\xi')=\arg\min_{\xi\in \Xcal}\|\xi-\xi'\|$ with $\|\cdot\|$ being the Euclidean norm of vector.\footnote{Note that most of the textbooks on stochastic algorithms such as \cite{AGbook,SDRbook} consider minimization problems for which the update rule should be replaced by $$\xi_{t+1}=\Pi_{\Xcal}\left(\xi_t-a_t \cdot \frac{1}{N}\sum_{n=1}^{N}\psi_{\xi_t}^{(n)}\right),$$ and the objective function is often assumed to be convex instead of concave.} As described in \cite[Chapter~5.9]{SDRbook}, for instance, if $\Xcal$ is convex, $U$ is strongly concave and differentiable with respect to $\xi$, and $\EE\left[ \|\psi_{\xi}\|_2^2\right]<\infty$ for any $\xi\in \Xcal$, then the estimate $\xi_t$ converges to the optimal $\xi^*$ with the mean squared error of $O(1/t)$.

There have been many variants of the classical Robbins-Monro algorithm proposed in the literature, notably such as Polyak-Ruppert averaging \cite{P90,R91} and stochastic counterpart of Nesterov's acceleration \cite{N83}. More recently, the idea of using not only the first moment of the gradient estimate but also its second moment to set the learning rates for individual design parameters in $\xi$ adaptively has been explored insensitively, especially in the machine learning community, see \cite{DHS11,TH12,KB15,RKK18}.

\section{Monte Carlo gradient estimation}\label{sec:mlmc}
Here we introduce two Monte Carlo estimators of the gradient $\nabla_\xi U(\xi)=\EE[\psi_{\xi}]$. Subsequently we propose an algorithm to efficiently search for optimal Bayesian experimental designs.

\subsection{Standard Monte Carlo}
The standard Monte Carlo method is one of the easiest and the most straightforward methods to approximate $\psi_\xi$. Let us estimate two expectations with respect to $\theta'$ by the Monte Carlo averages using common random samples of $\theta'$, respectively. Namely, for randomly chosen $\theta$ and $\epsilon$, let
\[ \psi_{\xi,M} :=  \frac{\nabla_{\xi}\rho(f_{\xi}(\theta,\epsilon)\mid \theta,\xi)}{\rho(f_{\xi}(\theta,\epsilon)\mid \theta,\xi)}- \frac{\nabla \varrho_{\xi,M}(\theta,\epsilon)}{\varrho_{\xi,M}(\theta,\epsilon)}, \]
with
\begin{align*}
    \varrho_{\xi,M}(\theta,\epsilon) & = \frac{1}{M}\sum_{m=1}^{M}\rho(f_{\xi}(\theta,\epsilon)\mid \theta'^{(m)},\xi), \\
    \nabla \varrho_{\xi,M}(\theta,\epsilon) & = \frac{1}{M}\sum_{m=1}^{M}\nabla_{\xi}\rho(f_{\xi}(\theta,\epsilon)\mid \theta'^{(m)},\xi),
\end{align*}
where $\theta'^{(1)},\ldots,\theta'^{(M)}$ are independent samples from the prior distribution $\pi_0$. More generally, for an importance distribution $q$ which may depend on the value of $f_{\xi}(\theta,\epsilon)$ or the outer random variables $\theta$ and $\epsilon$, we can consider
\begin{align}\label{eq:mc_variable} 
\psi_{\xi,M,q} :=  \frac{\nabla_{\xi}\rho(f_{\xi}(\theta,\epsilon)\mid \theta,\xi)}{\rho(f_{\xi}(\theta,\epsilon)\mid \theta,\xi)}- \frac{\nabla \varrho_{\xi,M,q}(\theta,\epsilon)}{\varrho_{\xi,M,q}(\theta,\epsilon)}, 
\end{align}
with
\begin{align*}
    \varrho_{\xi,M,q}(\theta,\epsilon) & = \frac{1}{M}\sum_{m=1}^{M}\frac{\rho(f_{\xi}(\theta,\epsilon)\mid \theta'^{(m)},\xi)\pi_0(\theta'^{(m)})}{q(\theta'^{(m)})},\\
    \nabla \varrho_{\xi,M,q}(\theta,\epsilon) & = \frac{1}{M}\sum_{m=1}^{M}\frac{\nabla_{\xi}\rho(f_{\xi}(\theta,\epsilon)\mid \theta'^{(m)},\xi)\pi_0(\theta'^{(m)})}{q(\theta'^{(m)})}, 
\end{align*}
where $\theta'^{(1)},\ldots,\theta'^{(M)}$ are independent samples from the distribution $q$. 

Although it holds from the linearity of expectation that 
\begin{align*}
    \EE\left[ \varrho_{\xi,M,q}(\theta,\epsilon) \mid \theta,\epsilon\right] & =\EE_{\theta'}\left[ \rho(f_{\xi}(\theta,\epsilon)\mid \theta',\xi)\right], \\
    \EE\left[ \nabla \varrho_{\xi,M,q}(\theta,\epsilon)\mid \theta,\epsilon\right] & =\EE_{\theta'}\left[ \nabla_{\xi}\rho(f_{\xi}(\theta,\epsilon)\mid \theta',\xi)\right],
\end{align*}
for any $M$, i.e., both the denominator and the numerator themselves are estimated without any bias, respectively, taking the ratio between these two yields 
\[ \EE[\psi_{\xi,M}], \EE[\psi_{\xi,M,q}] \neq \EE[\psi_\xi]=\nabla_\xi U(\xi)\] 
unless $q=\pi^{f_{\xi}(\theta,\epsilon)\mid \xi}$. This means that neither $\psi_{\xi,M}$ nor $\psi_{\xi,M,q}$ is an unbiased estimator of the gradient $\nabla_\xi U(\xi)$.

\subsection{Unbiased multilevel Monte Carlo}
Here we introduce an unbiased multilevel Monte Carlo estimator by using the debiasing technique from \cite{RG15} which itself is an extension of the multilevel Monte Carlo method due to Giles \cite{G08,G15}. Let us consider an increasing sequence $0<M_0<M_1<\ldots$ such that $M_\ell\to \infty$ as $\ell\to \infty$. Then the strong law of large numbers ensures that
\[ \PP\left[\lim_{\ell\to \infty}\psi_{\xi,M_\ell,q}=\psi_\xi\right]=1, \]
see for instance \cite[Theorem~9.2]{mcbook}, and so the following telescoping sum holds:
\[ \nabla_\xi U(\xi)=\EE[\psi_\xi]=\lim_{\ell\to \infty}\EE[\psi_{\xi,M_\ell,q}] = \EE[\psi_{\xi,M_0,q}] + \sum_{\ell=1}^{\infty}\EE[\psi_{\xi,M_\ell,q}-\psi_{\xi,M_{\ell-1},q}]. \]
More generally, suppose at this moment that we have a sequence of \emph{correction} random variables $\Delta \psi_{\xi,0},\Delta \psi_{\xi,1},\ldots$ such that $\EE[\Delta \psi_{\xi,0}] = \EE[\psi_{\xi,M_0,q}]$ and
\[ \EE[\Delta \psi_{\xi,\ell}] =\EE[\psi_{\xi,M_\ell,q}-\psi_{\xi,M_{\ell-1},q}] \quad \text{for $\ell>0$.}\]
Then it holds that
\begin{align}\label{eq:telescope}
    \nabla_\xi U(\xi)=\EE[\psi_\xi]=\sum_{\ell=0}^{\infty}\EE[\Delta \psi_{\xi,\ell}].
\end{align}
For any sequence of positive reals $w_0,w_1,\ldots$ such that $w_0+w_1+\cdots=1$, the expectation of the random variable
\[ \frac{\Delta \psi_{\xi,\ell}}{w_\ell} \]
with the index $\ell\geq 0$ being selected randomly with probability $w_\ell$, is equal to the gradient $\nabla_\xi U(\xi)$. In fact, it is easy to see that
\[ \EE\left[ \frac{\Delta \psi_{\xi,\ell}}{w_\ell}\right] = \sum_{\ell=0}^{\infty}\frac{\EE[\Delta \psi_{\xi,\ell}]}{w_\ell}w_\ell= \sum_{\ell=0}^{\infty}\EE[\Delta \psi_{\xi,\ell}] = \nabla_\xi U(\xi). \]
Therefore, for any number of outer samples $N\in \ZZ_{>0}$, 
\[ \frac{1}{N}\sum_{n=1}^{N}\frac{\Delta \psi_{\xi,\ell^{(n)}}}{w_{\ell^{(n)}}} \]
with $\ell^{(1)},\ldots,\ell^{(N)}$ being independent and randomly chosen with probability $w_\ell$ is an unbiased Monte Carlo estimator of $\nabla_\xi U(\xi)$.

Let $C_\ell$ denote the expected cost of computing $\Delta \psi_{\xi,\ell}$, which is proportional to $M_{\ell}$. In order for the random variable $\Delta \psi_{\xi,\ell}/w_\ell$ to have finite expected squared $\ell_2$-norm and finite expected computational cost, we must have
\begin{align}\label{eq:condition} \sum_{\ell=0}^{\infty}\frac{\EE[\|\Delta \psi_{\xi,\ell}\|_2^2]}{w_\ell}<\infty \quad \text{and}\quad \sum_{\ell=0}^{\infty}C_\ell w_\ell<\infty. \end{align}
Thus construction of such correction variables $\Delta \psi_{\xi,\ell}$ in conjunction with an associated sequence $w_0,w_1,\ldots$, which has not been discussed yet, becomes a central issue. 

\subsubsection{Naive construction}
Throughout this paper let us consider a geometric progression $M_\ell=M_02^{\ell}$ for some $M_0\in \ZZ_{\geq 0}$. Although it is possible to change the base of the progression to a general integer $b\geq 2$, we restrict ourselves to the case $b=2$ for simplicity of exposition. 

Probably the most straightforward form of the correction variables $\Delta \psi_{\xi,0},\Delta \psi_{\xi,1},\ldots$ is $\Delta \psi_{\xi,0}=\psi_{\xi,M_0,q}$ and
\[ \Delta \psi_{\xi,\ell} = \psi_{\xi,M_0 2^{\ell},q}-\psi_{\xi,M_0 2^{\ell-1},q} ,\] 
for $\ell>0$, where both $\psi_{\xi,M_0 2^{\ell-1},q}$ and $\psi_{\xi,M_0 2^{\ell},q}$ are given as in \eqref{eq:mc_variable} with $M=M_02^{\ell-1}$ and $M=M_02^{\ell}$, respectively. Here, instead of using mutually independent $M_02^{\ell-1}$ and $M_02^{\ell}$ samples on $\theta'$ to compute $\psi_{\xi,M_0 2^{\ell-1},q}$ and $\psi_{\xi,M_0 2^{\ell},q}$, respectively, a subset with size $M_02^{\ell-1}$ of the $M_02^{\ell}$ samples on $\theta'$ used to compute $\psi_{\xi,M_0 2^\ell,q}$, can be reused to compute $\psi_{\xi,M_0 2^{\ell-1},q}$ by the linearity of expectation. By doing so, it is expected that $\EE[\|\Delta \psi_{\xi,\ell}\|_2^2]$ is much smaller in magnitude than $\EE[\|\psi_{\xi,M_0 2^{\ell},q}\|_2^2]$ (or $\EE[\|\psi_{\xi,M_0 2^{\ell-1},q}\|_2^2]$).

However, it seems not possible that the order of $\EE[\|\Delta \psi_{\xi,\ell}\|_2^2]$ is better than $O(2^{-\ell})$. Recalling that $C_{\ell}\propto M_{\ell}\propto 2^{\ell}$, a faster decay of $\EE[\|\Delta \psi_{\xi,\ell}\|_2^2]$ is required to find a sequence of positive reals $w_0,w_1,\ldots$ which satisfies the condition \eqref{eq:condition}. We conjecture that a lower bound on $\EE[\|\Delta \psi_{\xi,\ell}\|_2^2]$ of order $2^{-\ell}$ exists for this naive construction.

\subsubsection{Antithetic construction}
Motivated by the MLMC literature \cite{GS14,BHR15,GG19,GHI20,HGGT20}, we address this issue by considering the following \emph{antithetic coupling} in this paper. A key ingredient here is that we can take two disjoint subsets with equal size $M_02^{\ell-1}$ of the $M_02^{\ell}$ samples on $\theta'$ used to compute $\psi_{\xi,M_0 2^\ell,q}$, which results in two independent realizations of $\psi_{\xi,M_0 2^{\ell-1},q}$, denoted by $\psi_{\xi,M_0 2^{\ell-1},q}^{(a)}$ and $\psi_{\xi,M_0 2^{\ell-1},q}^{(b)}$, respectively. To be more precise, for the independent samples $\theta'^{(1)},\ldots,\theta'^{(M_02^{\ell})}$ generated from the distribution $q$, we write
\begin{align*}
    \psi_{\xi,M_02^{\ell},q} & =  \frac{\nabla_{\xi}\rho(f_{\xi}(\theta,\epsilon)\mid \theta,\xi)}{\rho(f_{\xi}(\theta,\epsilon)\mid \theta,\xi)}- \frac{\nabla \varrho_{\xi,M_02^{\ell},q}(\theta,\epsilon)}{\varrho_{\xi,M_02^{\ell},q}(\theta,\epsilon)},\\
    \psi_{\xi,M_02^{\ell-1},q}^{(a)} & =  \frac{\nabla_{\xi}\rho(f_{\xi}(\theta,\epsilon)\mid \theta,\xi)}{\rho(f_{\xi}(\theta,\epsilon)\mid \theta,\xi)}- \frac{\nabla \varrho_{\xi,M_02^{\ell-1},q}^{(a)}(\theta,\epsilon)}{\varrho_{\xi,M_02^{\ell-1},q}^{(a)}(\theta,\epsilon)},\quad \text{and}\\
    \psi_{\xi,M_02^{\ell-1},q}^{(b)} & =  \frac{\nabla_{\xi}\rho(f_{\xi}(\theta,\epsilon)\mid \theta,\xi)}{\rho(f_{\xi}(\theta,\epsilon)\mid \theta,\xi)}- \frac{\nabla \varrho_{\xi,M_02^{\ell-1},q}^{(b)}(\theta,\epsilon)}{\varrho_{\xi,M_02^{\ell-1},q}^{(b)}(\theta,\epsilon)},
\end{align*}
where, for the second term of each, we have defined
\begin{align*}
    \varrho_{\xi,M_02^{\ell},q}(\theta,\epsilon) & = \frac{1}{M_02^{\ell}}\sum_{m=1}^{M_02^{\ell}}\frac{\rho(f_{\xi}(\theta,\epsilon)\mid \theta'^{(m)},\xi)\pi_0(\theta'^{(m)})}{q(\theta'^{(m)})},\\
    \nabla \varrho_{\xi,M_02^{\ell},q}(\theta,\epsilon) & = \frac{1}{M_02^{\ell}}\sum_{m=1}^{M_02^{\ell}}\frac{\nabla_{\xi}\rho(f_{\xi}(\theta,\epsilon)\mid \theta'^{(m)},\xi)\pi_0(\theta'^{(m)})}{q(\theta'^{(m)})}, \\
    \varrho_{\xi,M_02^{\ell-1},q}^{(a)}(\theta,\epsilon) & = \frac{1}{M_02^{\ell-1}}\sum_{m=1}^{M_02^{\ell-1}}\frac{\rho(f_{\xi}(\theta,\epsilon)\mid \theta'^{(m)},\xi)\pi_0(\theta'^{(m)})}{q(\theta'^{(m)})},\\
    \nabla \varrho_{\xi,M_02^{\ell-1},q}^{(a)}(\theta,\epsilon) & = \frac{1}{M_02^{\ell-1}}\sum_{m=1}^{M_02^{\ell-1}}\frac{\nabla_{\xi}\rho(f_{\xi}(\theta,\epsilon)\mid \theta'^{(m)},\xi)\pi_0(\theta'^{(m)})}{q(\theta'^{(m)})}, \\
    \varrho_{\xi,M_02^{\ell-1},q}^{(b)}(\theta,\epsilon) & = \frac{1}{M_02^{\ell-1}}\sum_{m=M_02^{\ell-1}+1}^{M_02^{\ell}}\frac{\rho(f_{\xi}(\theta,\epsilon)\mid \theta'^{(m)},\xi)\pi_0(\theta'^{(m)})}{q(\theta'^{(m)})},\\
    \nabla \varrho_{\xi,M_02^{\ell-1},q}^{(b)}(\theta,\epsilon) & = \frac{1}{M_02^{\ell-1}}\sum_{m=M_02^{\ell-1}+1}^{M_02^{\ell}}\frac{\nabla_{\xi}\rho(f_{\xi}(\theta,\epsilon)\mid \theta'^{(m)},\xi)\pi_0(\theta'^{(m)})}{q(\theta'^{(m)})},
\end{align*}
respectively.

Now a sequence of the correction random variables $\Delta \psi_{\xi,0},\Delta \psi_{\xi,1},\ldots$ is defined by $\Delta \psi_{\xi,0}=\psi_{\xi,M_0,q}$ and
\begin{align}
\Delta \psi_{\xi,\ell} & = \psi_{\xi,M_0 2^\ell,q}-\frac{\psi_{\xi,M_0 2^{\ell-1},q}^{(a)}+\psi_{\xi,M_0 2^{\ell-1},q}^{(b)}}{2} \label{eq:mlmc_correction} \\
& = \frac{1}{2}\left( \frac{\nabla \varrho_{\xi,M_02^{\ell-1},q}^{(a)}(\theta,\epsilon)}{\varrho_{\xi,M_02^{\ell-1},q}^{(a)}(\theta,\epsilon)}+\frac{\nabla \varrho_{\xi,M_02^{\ell-1},q}^{(b)}(\theta,\epsilon)}{\varrho_{\xi,M_02^{\ell-1},q}^{(b)}(\theta,\epsilon)}\right) - \frac{\nabla \varrho_{\xi,M_02^{\ell},q}(\theta,\epsilon)}{\varrho_{\xi,M_02^{\ell},q}(\theta,\epsilon)},\notag
\end{align}
for $\ell>0$. The difference between this antithetic construction and the naive construction is that $\psi_{\xi,M_0 2^{\ell-1},q}$ has been replaced by the mean of $\psi_{\xi,M_0 2^{\ell-1},q}^{(a)}$ and $\psi_{\xi,M_0 2^{\ell-1},q}^{(b)}$. This means that each of the $M_02^{\ell}$ samples is used exactly twice in antithetic construction: once in $\psi_{\xi,M_0 2^\ell,q}$ and once in either $\psi_{\xi,M_0 2^{\ell-1},q}^{(a)}$ or $\psi_{\xi,M_0 2^{\ell-1},q}^{(b)}$. For this novel version of $\Delta \psi_{\xi,\ell}$, the linearity of expectation ensures
\begin{align*}
    \EE\left[\Delta \psi_{\xi,\ell}\right] & = \EE\left[\psi_{\xi,M_0 2^\ell,q}\right]-\frac{1}{2}\left(\EE\left[\psi_{\xi,M_0 2^{\ell-1},q}^{(a)}\right]+\EE\left[\psi_{\xi,M_0 2^{\ell-1},q}^{(b)}\right]\right) \\
    & = \EE\left[\psi_{\xi,M_0 2^\ell,q}\right]-\frac{1}{2}\left(\EE\left[\psi_{\xi,M_0 2^{\ell-1},q}\right]+\EE\left[\psi_{\xi,M_0 2^{\ell-1},q}\right]\right) \\
    & = \EE[\psi_{\xi,M_0 2^\ell,q}-\psi_{\xi,M_0 2^{\ell-1},q}],
\end{align*} 
so that it fits with the telescoping sum representation \eqref{eq:telescope} of the gradient $\nabla_\xi U(\xi)$. Despite the distinction being subtle, we will show that the antithetic construction has better properties than the naive construction. Hereafter, $\Delta \psi_{\xi,\ell}$ refers to the antithetic construction given in \eqref{eq:mlmc_correction}.

It is clear that the cost $C_\ell$ to compute $\Delta \psi_{\xi,\ell}$ is proportional to $2^{\ell}$, and also that the following \emph{antithetic} properties hold for $\Delta \psi_{\xi,\ell}$:
\begin{align}\label{eq:antithetic}
\begin{split}
\varrho_{\xi,M_02^{\ell},q}(\theta,\epsilon) & = \frac{1}{2}\left( \varrho_{\xi,M_02^{\ell-1},q}^{(a)}(\theta,\epsilon)+\varrho_{\xi,M_02^{\ell-1},q}^{(b)}(\theta,\epsilon)\right),\quad \text{and}\\
\nabla \varrho_{\xi,M_02^{\ell},q}(\theta,\epsilon) & = \frac{1}{2}\left( \nabla \varrho_{\xi,M_02^{\ell-1},q}^{(a)}(\theta,\epsilon)+\nabla \varrho_{\xi,M_02^{\ell-1},q}^{(b)}(\theta,\epsilon)\right) ,
\end{split}
\end{align}
which play a crucial role in showing that this antithetic construction achieves a faster decay rate of $\EE[\|\Delta \psi_{\xi,\ell}\|_2^2]$ than the naive construction, making it possible to find a sequence of positive reals $w_0,w_1,\ldots$ which satisfies the condition \eqref{eq:condition}. The following claim is the main theoretical result of this paper.
\begin{theorem}\label{thm:mlmc_conv}
Assume that 
\[ \sup_{\theta,\theta',\epsilon} \left\| \nabla_{\xi}\log \rho(f_{\xi}(\theta,\epsilon)\mid \theta',\xi)\right\|_{\infty} <\infty, \]
and that there exists $u> 2$ such that
\[ \EE_{\theta\sim \pi_0,\theta'\sim q,\epsilon} \left[\left| \frac{\rho(f_{\xi}(\theta,\epsilon)\mid \theta',\xi)\pi_0(\theta')}{ \rho(f_{\xi}(\theta,\epsilon)\mid \xi)q(\theta')}\right|^u \right] <\infty. \]
Then the following holds true:
\begin{enumerate}
\item For a fixed $\ell$, we have
\[ \EE[\|\Delta \psi_{\xi,\ell}\|_2^2]=O(2^{-\beta \ell}) \quad \text{with}\quad  \beta=\frac{\min\left(u,4\right)}{2}. \]
\item In order to have \eqref{eq:condition}, it suffices to choose $w_{\ell} \propto 2^{-\tau\ell}$ with $1<\tau<\beta$.
\end{enumerate}
\end{theorem}
\noindent We postpone the proof of the theorem to Appendix~\ref{app:proof}. 

\begin{remark}\label{rem:mlmc_conv}
It follows from the first item of Theorem~\ref{thm:mlmc_conv} that
\[ \EE[\|\Delta \psi_{\xi,\ell}\|_2]=O(2^{-\ell}), \]
for a fixed $\ell$. Using this property, the bias of the standard Monte Carlo estimator $\psi_{\xi,M_02^L,q}$ with $M=M_02^L$ inner samples is bounded as
\begin{align*}
\left\| \nabla_\xi U(\xi) - \EE[\psi_{\xi,M_02^L,q}]\right\|_2 & = \left\|\sum_{\ell=L+1}^{\infty}\EE[\Delta \psi_{\xi,\ell}]\right\|_2 \leq \sum_{\ell=L+1}^{\infty}\EE\left[\left\|\Delta \psi_{\xi,\ell}\right\|_2 \right] \\
& = O(2^{- L}) = O(M^{-1}).
\end{align*}
This means that, for small $M$, the standard Monte Carlo estimator may lead to a wrong trajectory of an experimental design in stochastic gradient-biased optimization and the resulting design will not be close to optimal.
\end{remark}

\subsection{Unbiased MLMC stochastic optimization}
Finally we arrive at our proposal of a stochastic algorithm to search for an optimal Bayesian experimental design $\xi^*\in \Xcal$ as summarized in Algorithm~\ref{alg:mlmc_sgd}. Here we note that Algorithm~\ref{alg:mlmc_sgd} assumes that the conditions appearing in Theorem~\ref{thm:mlmc_conv} hold for any $\xi\in \Xcal$ with a common value of $u$. Given additional assumptions that the domain $\Xcal$ is convex and that $U$ is strongly concave and differentiable with respect to $\xi$, most of the stochastic gradient-based optimization algorithms have a theoretical guarantee that $\xi_t$ converges to the optimal $\xi^*\in \Xcal$ with some decay rate, typically with the mean square error of $O(1/t)$ as mentioned in Section~\ref{subsec:gradient}.
\begin{algorithm}
\caption{Unbiased MLMC stochastic optimization}
\label{alg:mlmc_sgd} 
For a given $1<\tau <\beta$, set $w_0,w_1,\ldots>0$ such that $w_0+w_1+\cdots=1$ and $w_{\ell} \propto 2^{-\tau \ell}$. For the feasible set $\Xcal$, initialize $\xi_0\in \Xcal$ and $t=0$. For $N\in \ZZ_{>0}$, do the following:
\begin{enumerate}
\item Choose $\ell^{(1)},\ldots,\ell^{(N)}\in \ZZ_{\geq 0}$ independently and randomly with probability $w_\ell$.
\item Compute an unbiased MLMC estimate of the gradient $\nabla_{\xi}U$ at $\xi=\xi_t$:
\[ \frac{1}{N}\sum_{n=1}^{N}\frac{\Delta \psi_{\xi_t,\ell^{(n)}}}{w_{\ell^{(n)}}}. \]
\item Apply a stochastic gradient-based algorithm to get $\xi_{t+1}$:
\[ \xi_{t+1}=g_t\left(\xi_t, \frac{1}{N}\sum_{n=1}^{N}\frac{\Delta \psi_{\xi_t,\ell^{(n)}}}{w_{\ell^{(n)}}}  \right).\]
\item Check whether a certain stopping criterion is satisfied. If yes, stop the iteration. Otherwise, go to Step~1 with $t\leftarrow t+1$.
\end{enumerate}
\end{algorithm}

As in \cite{BDELT18,GHI20,CDELT20}, using Laplace approximation-based importance distribution for $q$ helps not only reduce the expected squared $\ell_2$-norm of the Monte Carlo gradient estimator but also avoid numerical instability coming from concentrated posterior measures of $\theta'$ given $f_{\xi}(\theta,\epsilon)$. We also refer to \cite{SSWxx} for some theoretical analyses on the Laplace approximation. 

\section{Numerical experiments}\label{sec:numerics}
Here, we conduct numerical experiments on two example problems in Bayesian experimental design. The first example is aimed at verifying our proposed algorithm by using a simple test problem. Then, in order to see practical performance of our algorithm, we consider a PK model used in \cite{RDTP14} for our second example. The Python code used in our experiments is available from \url{https://github.com/Goda-Research-Group/MLMC_stochastic_gradient}.

\subsection{Simple test case}
Let $\theta=(\theta_1,\theta_2)\in \RR_{>0}^2$ with $\theta_1,\theta_2\overset{\mathrm{iid}}{\sim} \mathrm{lognormal}(\mu,\sigma_0^2).$ For an experimental design $\xi\in \RR_{>0}$, let an observation $Y=(Y_1,Y_2)\in \RR_{>0}^2$ follow 
\begin{align*}
    Y_1\mid \theta,\xi & \sim \mathrm{lognormal}(g(\xi)\log \theta_1,\sigma_{\epsilon}^2),\\
    Y_2\mid \theta,\xi & \sim \mathrm{lognormal}(h(\xi)\log \theta_2,\sigma_{\epsilon}^2),
\end{align*}
for some functions $g$ and $h$. This is equivalent to consider the following forward model:
\begin{align*}
    Y_1 & = e^{g(\xi)\log \theta_1+\sigma_{\epsilon}\epsilon_1}, \\
    Y_2 & = e^{h(\xi)\log \theta_2+\sigma_{\epsilon}\epsilon_2},
\end{align*}
for $\epsilon_1,\epsilon_2\overset{\mathrm{iid}}{\sim} N(0,1)$ independently of $\theta$ and $\xi$, which is obviously a special case of \eqref{eq:forward}. The expected information gain for a given $\xi$ is analytically calculated as
\[ U(\xi) = \frac{1}{2}\log \left( \left(g(\xi)\right)^2\frac{\sigma_0^2}{\sigma_{\epsilon}^2}+1\right)\left( \left(h(\xi)\right)^2\frac{\sigma_0^2}{\sigma_{\epsilon}^2}+1\right). \]
Also, applying Jensen's inequality to \eqref{eq:eig2}, we see that $U(\xi)$ is bounded above by
\begin{align*}
    U(\xi)\leq \tilde{U}(\xi) & := \EE_{\theta, \epsilon}\left[ \log \rho(f_{\xi}(\theta,\epsilon)\mid \theta,\xi)\right]- \EE_{\theta, \theta',\epsilon}\left[ \log \rho(f_{\xi}(\theta,\epsilon)\mid \theta',\xi) \right] \\
    & \, = \left(g(\xi)\right)^2\frac{\sigma_0^2}{\sigma_{\epsilon}^2}+\left(h(\xi)\right)^2\frac{\sigma_0^2}{\sigma_{\epsilon}^2}.
\end{align*}
Here we note that the standard Monte Carlo gradient estimator with $M=1$ inner sample from the prior distribution, i.e., $\psi_{\xi,1}$, is nothing but an unbiased estimator of $\nabla_{\xi}\tilde{U}(\xi)$. Therefore, as long as 
\[ \xi^*=\arg\max_{\xi\in \RR_{>0}}U(\xi) \neq \arg\max_{\xi\in \RR_{>0}}\tilde{U}(\xi)\]
holds, stochastic gradient-based optimization based on $\psi_{\xi,1}$ will not converge to the optimal design $\xi^*$. In our experiments below, let $\mu=0$, $\sigma_0=\sigma_{\epsilon}=1$, 
\[ g(\xi)=e^{-\xi^2/2}\quad \text{and}\quad h(\xi)=\sqrt{\frac{3}{2}\left(1-e^{-\xi^2}\right)}.\]
Fig.~\ref{fig:test-eig} compares $U$ and $\tilde{U}$ as functions of $\xi$ for this setting. The optimal design which maximizes $U$ is given by $\xi^*=\sqrt{\log 3}\approx 1.048\ldots$ and we can observe that $U$ is concave around $\xi^*$. On the other hand, its upper bound $\tilde{U}$ is a strictly monotone increasing function and its supremum attains for $\xi\to \infty$.
\begin{figure}
    \centering
    \includegraphics[width=0.4\textwidth]{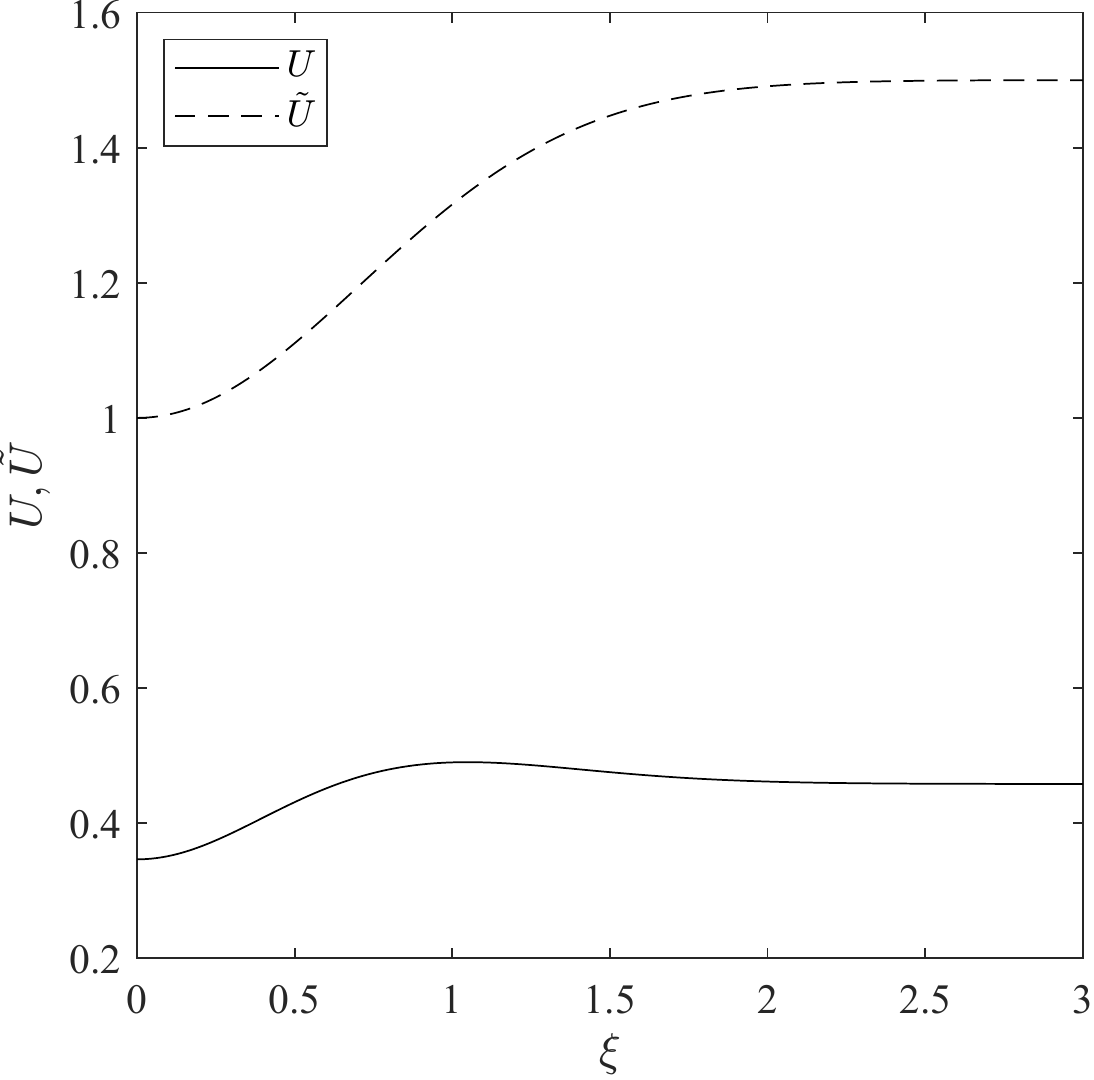}
    \caption{The expected information gain $U$ and its upper bound $\tilde{U}$ for the test case.}
    \label{fig:test-eig}
\end{figure}

Throughout this subsection, we do not use any importance sampling for the unbiased MLMC estimator of $\nabla_{\xi}U$ and set $M_0$, the number of level $0$ inner samples, to $1$. The left panel of Fig.~\ref{fig:test-mlmc_convergence} shows the convergence behavior of the MLMC correction variables $\Delta \psi_{\xi,\ell}$ at $\xi=1.5$. Here the mean squares (expected squared $\ell_2$-norms) of $\psi_{\xi,M_\ell}$ and $\Delta \psi_{\xi,\ell}$ are plotted on a $\log_2$ scale as functions of the level $\ell$, where the means are estimated empirically by using $10^5$ i.i.d.\ samples at each level. While the mean square of $\psi_{\xi,M_\ell}$ takes an almost constant value for $\ell> 4$, that of $\Delta \psi_{\xi,\ell}$ decreases geometrically as the level increases. The linear regression of the data for the range $1\leq \ell\leq 10$ provides an estimation of $\beta$ as $1.64$, which agrees well with the theoretical result in Theorem~\ref{thm:mlmc_conv}. As shown in the right panel of Fig.~\ref{fig:test-mlmc_convergence}, a similar convergence behavior of the MLMC correction variables $\Delta \psi_{\xi,\ell}$ can be observed at the optimal design $\xi=\xi^*=\sqrt{\log 3}$, where $\beta$ is estimated as $1.63$.

\begin{figure}
    \centering
    \includegraphics[width=0.4\textwidth]{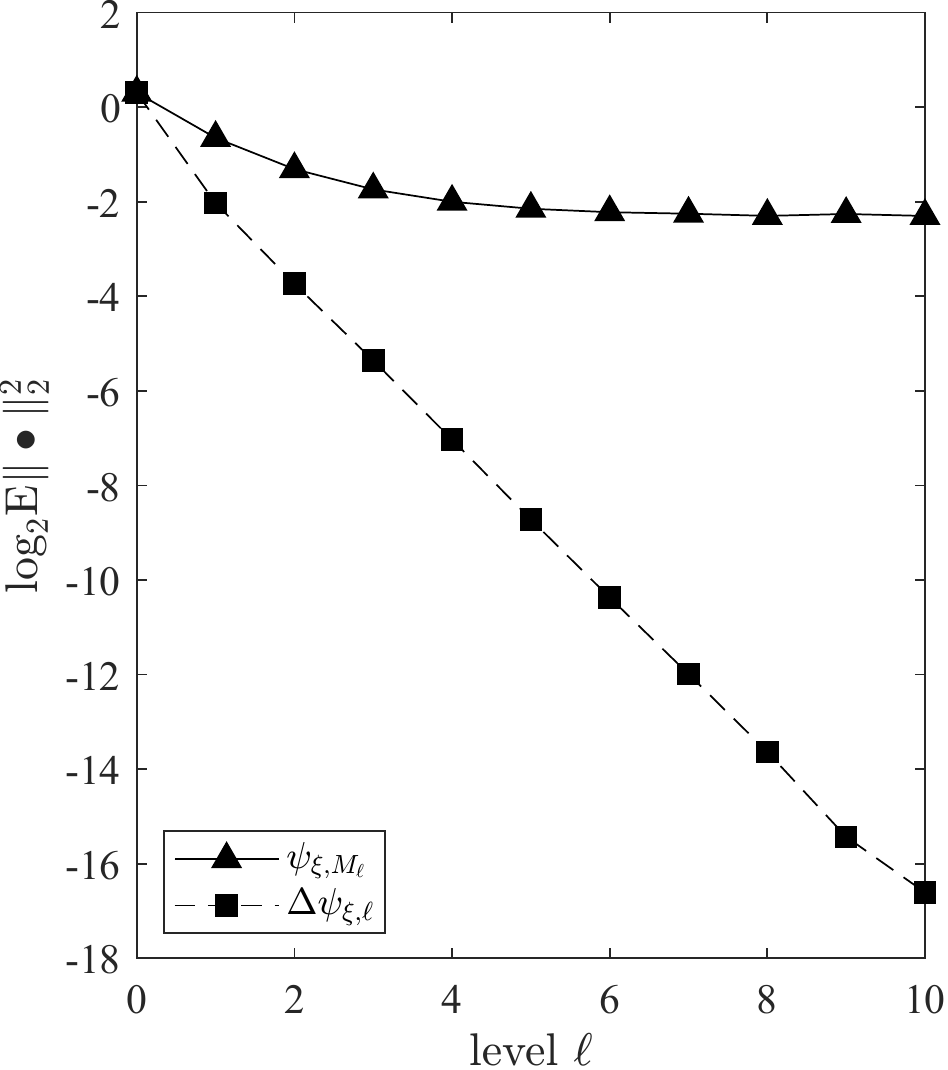}
    \includegraphics[width=0.4\textwidth]{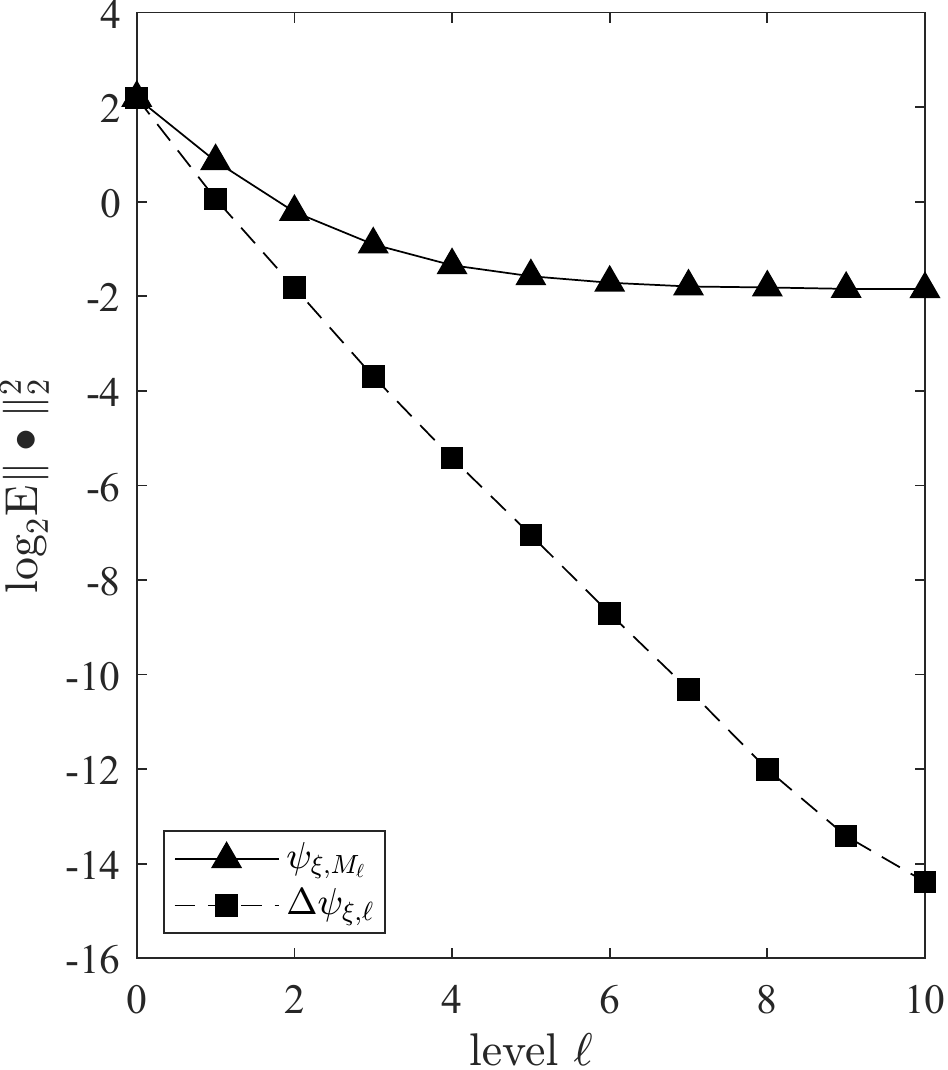}
    \caption{The mean squares of the variables $\psi_{\xi,M_\ell}$ and $\Delta \psi_{\xi,\ell}$ for the test case at $\xi=1.5$ (left) and at $\xi=\xi^*=\sqrt{\log 3}$ (right).}
    \label{fig:test-mlmc_convergence}
\end{figure}

Such a fast geometric decay of the correction variables $\Delta \psi_{\xi,\ell}$ justifies us to apply Algorithm~\ref{alg:mlmc_sgd} to search for the optimal design $\xi^*$. In order to randomly choose the level $\ell$, we set $\tau=1.5$ and $w_\ell = 2^{-3\ell/2}(1-2^{-3/2})$. This implies that the expected number of inner samples used in the MLMC estimator is given by
\[ \sum_{\ell=0}^{\infty}2^{\ell}w_{\ell} = (1-2^{-3/2})\sum_{\ell=0}^{\infty}2^{-\ell/2}=\frac{1-2^{-3/2}}{1-2^{-1/2}}\approx 2.21. \]
For comparison, we also consider the standard Monte Carlo estimators $\psi_{\xi,M}$ for the gradient of the expected information gain with various values of $M=1,2,4,\ldots,64$ within stochastic gradient descent. We fix the number of outer samples $N$ to $2000$ throughout all the iteration steps for all the estimators. We use the Robbins-Monro algorithm with Polyak-Ruppert averaging and the learning rates $\alpha_t=5/(t+1)$ as a stochastic descent algorithm, and as the computational cost is proportional to the number of inner samples, we set the maximum iteration steps $T$ to $\lfloor 10^7/M\rfloor$, the largest integer less than or equal to $10^7/M$. Although the number of inner samples is a random variable for the MLMC estimator, we simply set $T$ to $\lfloor 10^7/2.21\rfloor$. The initial design candidate at $t=0$ is given by $\xi_0=1.5$ and the feasible set $\Xcal$ is set to $\RR_{>0}$. Hence the maximum increment of the expected information gain is $U(\sqrt{\log 3})-U(1.5)\approx 0.0148$. For each gradient estimator, we conduct 10 independent runs and compute the average of the distance $\|\xi_t-\xi^*\|_2^2$ and its standard error for all the iteration steps, which correspond to the line and the shaded area of Fig.~\ref{fig:test-trajectory}, respectively.

Fig.~\ref{fig:test-trajectory} shows the convergence behaviors of the estimated experimental design $\xi_t$ for the considered estimators of the gradient $\nabla_{\xi}U$. Note here that the horizontal axis is given by $M\times t$ as a measure of the total computational cost (here again, we simply let $M=2.21$ for the MLMC estimator) and both axes use the logarithmic scales. As expected, the standard Monte Carlo estimator with $M=1$ leads to larger values of $\xi_t$ which make $\tilde{U}$ large, so that the search goes in wrong direction. Even for $M=2$, the situation is not improved so much and the experimental design $\xi_t$ remains almost the same throughout the iterations. For larger values of $M$, the standard Monte Carlo estimator works in the early stages, making the distance $\|\xi_t-\xi^*\|_2^2$ small. However, after some iteration steps, the estimate $\xi_t$ converges to some point away from the optimal $\xi^*$. Although it is natural that such bias can be reduced simply by increasing $M$, a proper choice of $M$ in practical applications is far from trivial since larger $M$ means a larger computational cost and the bias seems extremely hard to estimate in advance. This is exactly the point where the unbiased MLMC estimator can help. As the black line shows, the distance $\|\xi_t-\xi^*\|_2^2$ decreases consistently from the early stage and overtakes the standard Monte Carlo estimators with fixed $M$, leading to a better estimate of the optimal experimental design. The linear regression of the data for the whole range $0< \log_{10}(Mt)\leq 7$ shows that the estimate $\xi_t$ converges to $\xi^*$ with the mean squared error of order $t^{-1.12}$ approximately, which is almost consistent with the standard stochastic optimization theory \cite[Chapter~5.9]{SDRbook}. A slightly faster decay of the standard Monte Carlo estimators with $M\geq 8$ in the early stages could be because that they estimate the gradients of biased objective functions which are steeper than the gradient of $U(\xi)$ around the initial estimate $\xi_0=1.5$ in this case.

\begin{figure}
    \centering
    \includegraphics[width=0.8\textwidth]{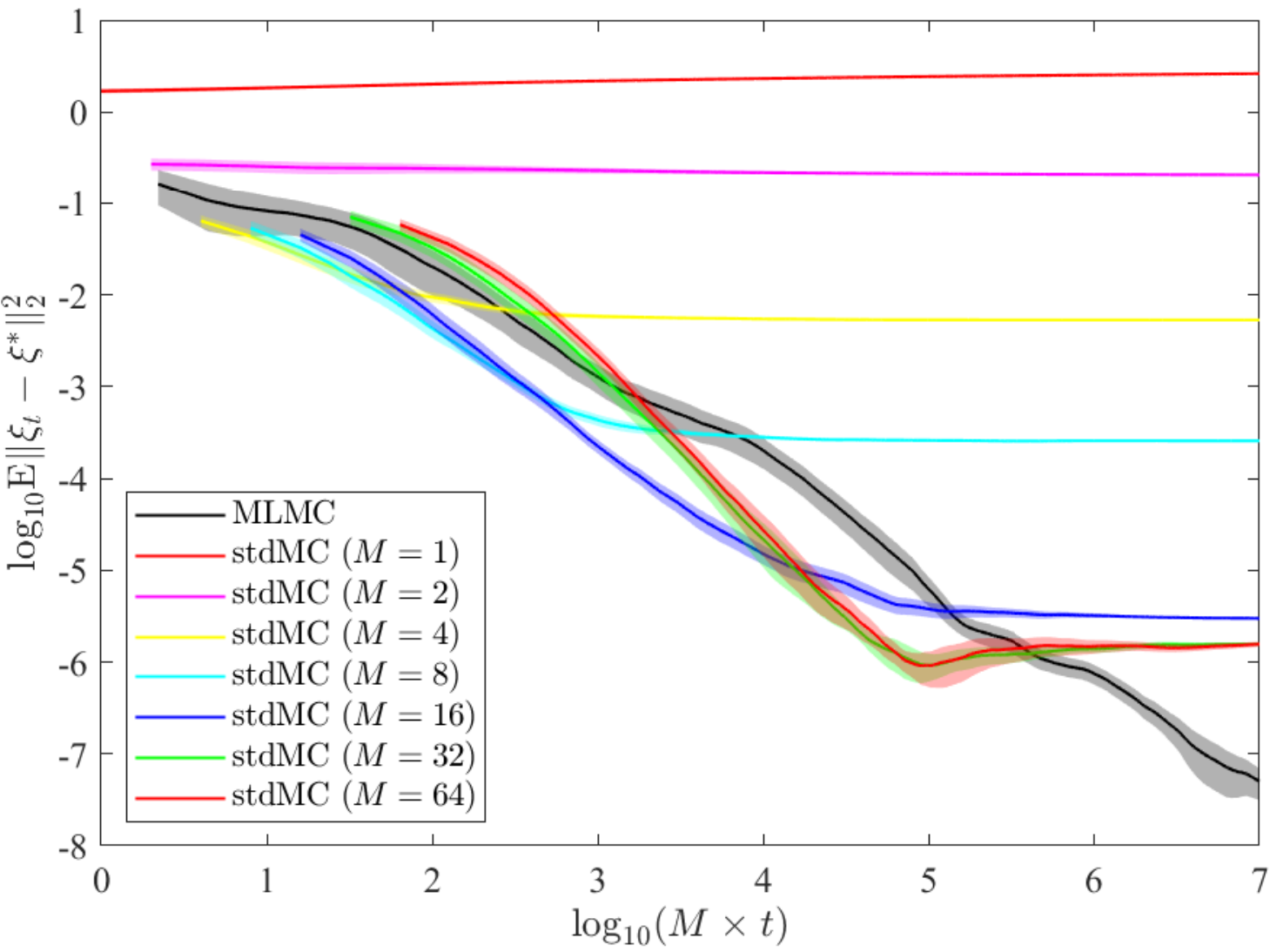}
    \caption{The convergence of the estimated experimental design $\xi_t$ to the optimal $\xi^*$ for various Monte Carlo estimators of the gradient $\nabla_{\xi}U$. For each estimator, the line and the shaded area represent the average and its standard error estimated from 10 independent runs, respectively.}
    \label{fig:test-trajectory}
\end{figure}

\subsection{Pharmacokinetic model}
Let us consider a PK design problem introduced in \cite{RDTP14}. Suppose that a drug with a fixed dose $D=400$ is administrated to subjects at time $\Tcal=0$. In order to reduce the uncertainty about a set of PK parameters, which affect the absorption, distribution and the elimination of the drug in the subjects’ body, it would be helpful to take blood samples of the subjects at several different times and to measure the concentration of the drug in the samples. Blood samples are assumed to be taken 15 times at $\Tcal=\xi^{(1)}, \ldots, \xi^{(15)}$ hours after the drug administration. Given the set of 15 drug concentration measurements, it is expected that the uncertainty of PK parameters of interest $\theta$ can be reduced. Our objective here is to optimize sampling times $\xi=(\xi^{(1)}, \ldots, \xi^{(15)})\in \RR_{\geq 0}^{15}$ such that the expected information gain brought from blood sampling is maximized. 

Let $\theta=(\log k_a, \log k_e,  \log V)\in \RR^3$ where $k_a$ represents the first-order absorption rate constant, $k_e$ does the first-order elimination rate constant and $V$ does the volume of distribution. Following \cite{RDTP14}, assume that the drug concentration of blood sample taken at time $\Tcal\geq 0$ is described as
\[ Y_{\Tcal} = \frac{Dk_a}{V(k_a-k_e)}\left( e^{-k_e \Tcal}-e^{-k_a \Tcal}\right)\left(1+\epsilon_1\right) +\epsilon_2=: g_{\Tcal}(\theta,\epsilon),  \]
with $\epsilon=(\epsilon_1,\epsilon_2)$, where $\epsilon_1$ and $\epsilon_2$ represent the multiplicative and additive Gaussian noises, respectively. Then our forward model is given by
\[ Y = (Y_{\xi^{(1)}},\ldots,Y_{\xi^{(15)}}) =\left(g_{\xi^{(1)}}(\theta,\epsilon_{\xi^{(1)}}),\ldots, g_{\xi^{(15)}}(\theta,\epsilon_{\xi^{(15)}})\right) \in \RR^{15},\]
where $\epsilon_{\xi^{(1)}},\ldots,\epsilon_{\xi^{(15)}}$ are assumed mutually independent and follow the same bi-variate normal distribution 
\[ \epsilon_{\xi^{(j)}} \sim N\left( \begin{pmatrix} 0 \\ 0\end{pmatrix}, \begin{pmatrix} 0.01 & 0 \\ 0 & 0.1\end{pmatrix}\right).\] The input random variables in $\theta$ are assumed independent and the corresponding probability distributions are given by $\log k_a \sim N(0, 0.05),\log k_e \sim N(\log(0.1), 0.05)$ and $\log V \sim N(\log(20), 0.05)$, respectively. This means that the prior information entropy of $\theta$ is equal to $3\log(\sqrt{2\pi e \times 0.05})\approx -0.2368.$ Moreover, the likelihood function is given by the product of $\rho(g_{\xi^{(j)}}(\theta,\epsilon_{\xi^{(j)}})\mid \theta',\xi^{(j)})$
that can be computed explicitly by following Example~\ref{exm:mix}.

In this setting the posterior distribution of $\theta$ given $Y$ cannot be computed analytically. In order to reduce the expected squared $\ell_2$-norm of the unbiased MLMC estimator of the gradient $\nabla_{\xi}U$, we use Laplace approximation-based importance sampling. Since not only the additive noise but also the multiplicative noise are included in the forward model, we consider a simple modification of the original method in \cite{LSTW13} as follows. Let us write
\[ \overline{g_{\Tcal}}(\theta)=\frac{Dk_a}{V(k_a-k_e)}\left( e^{-k_e \Tcal}-e^{-k_a \Tcal}\right)\quad \text{and}\quad \overline{g_{\xi}}(\theta)=\left(\overline{g_{\xi^{(1)}}}(\theta),\ldots,\overline{g_{\xi^{(15)}}}(\theta) \right). \]
Then, for the data $Y$ generated conditionally on the known value of $\theta=\theta^*$ from  the forward model, we approximate the posterior distribution $\pi^{Y\mid \xi}(\theta)$ by a Gaussian distribution $N(\hat{\theta},\hat{\Sigma})$ with
\begin{align*}
    \hat{\theta} & = \theta^*-\left(J(\theta^*)^{\top}\Sigma_{\epsilon}^{-1}J(\theta^*)+H(\theta^*)^{\top}\Sigma_{\epsilon}^{-1}E-\nabla_{\theta}\nabla_{\theta}\log \pi_0(\theta^*) \right)^{-1}J(\theta^*)^{\top}\Sigma_{\epsilon}^{-1}E, \\
    \hat{\Sigma} & = \left( J(\hat{\theta})^{\top}\Sigma_{\epsilon}^{-1}J(\hat{\theta})-\nabla_{\theta}\nabla_{\theta}\log \pi_0(\hat{\theta})\right)^{-1}.
\end{align*}
Here $J$ and $H$ denote the Jacobian and Hessian of $-\overline{g_{\xi}}$, respectively, that is, $J(\theta)=-\nabla_{\theta}\overline{g_{\xi}}(\theta)$ and $H(\theta)=-\nabla_{\theta}\nabla_{\theta}\overline{g_{\xi}}(\theta)$. Also we write $E:=Y^{\top}-\overline{g_{\xi}}(\theta^*)^{\top}$ and
\[ \Sigma_{\epsilon} = \diag\left( 0.01\left(\overline{g_{\xi^{(1)}}}(\theta)\right)^2 +0.1,\ldots, 0.01\left(\overline{g_{\xi^{(15)}}}(\theta)\right)^2 +0.1\right). \]
We use this $N(\hat{\theta},\hat{\Sigma})$ as an importance distribution $q$. The only difference from the one in \cite{LSTW13} is that the matrix $\Sigma_{\epsilon}$ depends on the mean response $\overline{g_{\xi}}(\theta)$ due to the multiplicative noise in our setting. Although a first-order approximation argument similar to \cite{LSTW13} might be possible and lead to different forms of $\hat{\theta}$ and $\hat{\Sigma}$, such a detailed analysis on the Laplace approximation is beyond the scope of this paper.

In order to search for optimal design parameters $\xi=(\xi^{(1)}, \ldots, \xi^{(15)})$, we do not represent them by a smaller number of parameters as considered in \cite{RDTP14}, but instead we optimize them directly. We set a design at the initial iteration step $t=0$ to equi-spaced times $\xi_0=(1,2,\ldots,15)$. In Algorithm~\ref{alg:mlmc_sgd}, we fix $M_0=1$, set $w_0 = 0.9$ and $w_\ell \propto 2^{-3\ell/2}$ for $\ell\geq 1$ such that they are summed up to 1, and set the number of outer samples to $N=2000$ at each iteration step. This implies that the expected number of inner samples used in the MLMC estimator is given by 
\[ \sum_{\ell=0}^{\infty}2^{\ell}w_{\ell} = \frac{9}{10}+\frac{2^{3/2}-1}{10}\sum_{\ell=1}^{\infty}2^{-\ell/2}\approx 1.34. \]
We use the AMSGrad optimizer with constant learning rate $\alpha_t=0.004$ and exponential moving average parameters $\beta_1=0.9, \beta_2=0.999$ as a stochastic descent algorithm, and set the maximum iteration steps $T$ to $10000$ as a stopping criterion. The feasible domain $\Xcal$ is restricted to $[0, 24]^{15}$. For comparison, we also consider the standard (biased) Monte Carlo estimator for the gradient $\nabla_{\xi}U$ with a fixed number of inner samples $M=1$ and the Laplace approximation-based importance sampling within stochastic gradient descent. As expected from the numerical results shown in \cite{CDELT20}, the Laplace approximation-based importance sampling helps reduce the bias of the Monte Carlo estimator significantly even for $M=1$.

\begin{figure}
(a) stdMC
\begin{center}
\includegraphics[width=0.8\textwidth]{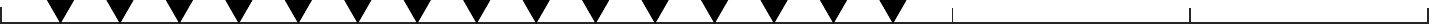}\\[5pt]
\includegraphics[width=0.8\textwidth]{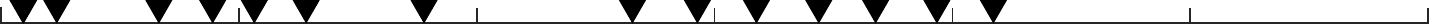}\\[5pt]
\includegraphics[width=0.8\textwidth]{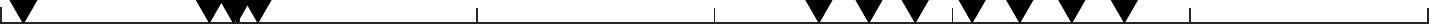}\\[5pt]
\includegraphics[width=0.8\textwidth]{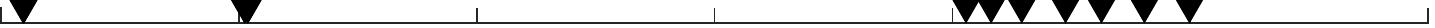}\\[5pt]
\includegraphics[width=0.8\textwidth]{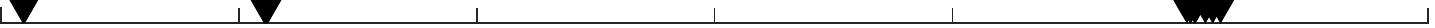}\\[5pt]
\includegraphics[width=0.8\textwidth]{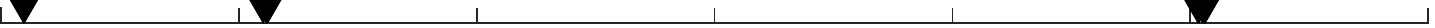}\\
{\footnotesize $\underbrace{(0.38, 0.38, 0.39)}_{3}\quad \underbrace{(4.41, 4.41, 4.46, 4.47, 4.49)}_{5}\quad \underbrace{(20.13, 20.15, 20.20, 20.22, 20.24, 20.25, 20.25)}_{7}$}\\[5pt]
\end{center}

(b) MLMC
\begin{center}
\includegraphics[width=0.8\textwidth]{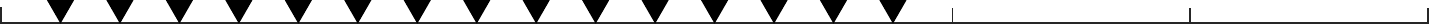}\\[5pt]
\includegraphics[width=0.8\textwidth]{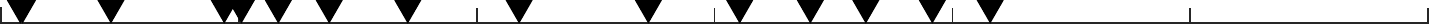}\\[5pt]
\includegraphics[width=0.8\textwidth]{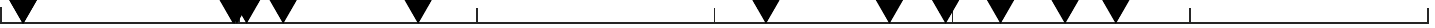}\\[5pt]
\includegraphics[width=0.8\textwidth]{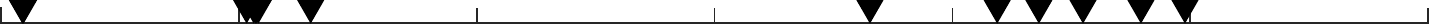}\\[5pt]
\includegraphics[width=0.8\textwidth]{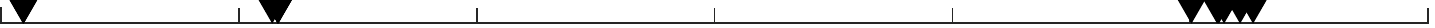}\\[5pt]
\includegraphics[width=0.8\textwidth]{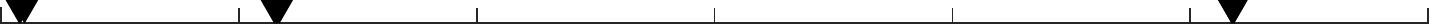}\\
{\footnotesize $\underbrace{(0.31, 0.33, 0.40)}_{3}
\quad \underbrace{(4.59, 4.63, 4.64, 4.64, 4.64, 4.68)}_{6}\quad \underbrace{(20.70, 20.70, 20.70, 20.71, 20.73, 20.74)}_{6}$}
\end{center}
    \caption{Design parameters $(\xi_1,\ldots,\xi_{15})$ within the interval $[0, 24]$ during the optimization process for a single run at the iteration steps $t=0, 100, 500, 1000, 5000, 10000$ (in descending order): (a) the result for stdMC and (b) the result for MLMC. The resulting design is shown in detail respectively at the bottom.}
    \label{fig:blood_times_history}
\end{figure}

Fig.~\ref{fig:blood_times_history} shows the set of design parameters $\xi=(\xi^{(1)}, \ldots, \xi^{(15)})$ obtained at the iteration steps $t=0, 100, 500, 1000, 5000, 10000$ for a single run. The overall convergence behaviors both for the standard Monte Carlo estimator and the MLMC estimator look quite similar to each other. That is, the allocations of 15 sampling times become irregular at the earlier steps compared to the initially equi-spaced design, but then some of sampling times gradually get quite close to each other, ending up with three well-separated clusters. It is interesting to see that stochastic gradient-based optimization naturally finds such so-called \emph{replicate} design that is often considered in the PK applications \cite{RDP15,RDTP14}. Looking into the details, there is a difference between the resulting designs obtained by the standard Monte Carlo estimator and the MLMC estimator. For the standard Monte Carlo estimator, the number of sampling times allocated to each cluster is $3, 5, 7$ (from earlier one to later one), respectively, whereas the corresponding number is $3, 6, 6$, respectively, for the MLMC estimator. These allocations of sampling times are consistent among 10 independent runs for both the estimators. The average sampling time (with its standard deviation) within each cluster, estimated from 10 independent runs, is $0.385\ (0.003), 4.442\ (0.008), 20.202\ (0.006)$ for the standard Monte Carlo estimator, and is $0.367\ (0.010), 4.652\ (0.018), 20.699\ (0.017)$ for the MLMC estimator. The two-sample Wilcoxon test yields the p-value about $10^{-5}$ for all of the three clusters, which supports that the differences between the centers of the clusters obtained by the two estimators are statistically significant.

Fig.~\ref{fig:pk-variance} shows the convergence behaviors of the MLMC correction variables $\Delta \psi_{\xi,\ell}$ at the iteration steps $t=0, T/2, T$ for a single run. Similarly to Fig.~\ref{fig:test-mlmc_convergence}. the mean squares (expected squared $\ell_2$-norms) of $\psi_{\xi,M_\ell}$ and $\Delta \psi_{\xi,\ell}$ are plotted on a $\log_2$ scale as functions of the level $\ell$, where the means are estimated empirically by using $10^5$ i.i.d.\ samples at each level. While the mean square of $\psi_{\xi,M_\ell}$ takes an almost constant value for $\ell> 4$, that of $\Delta \psi_{\xi,\ell}$ decreases geometrically as the level increases. The linear regression of the data for the range $1\leq \ell\leq 10$ provides estimations of $\beta$ as $0.80, 1.36, 1.47$, respectively. The result on the case $\beta\leq 1$ is not covered by Theorem~\ref{thm:mlmc_conv}, and in such a case, we do not have a right choice of $w_{\ell}$ which leads to both finite expected cost and finite expected squared $\ell_2$-norm. Further theoretical investigation is needed to address this issue. On the other hand, the result $\beta>1$ for the steps $t=T/2, T$ is as expected from our theoretical result. Nonetheless, our choice $w_\ell \propto 2^{-3\ell/2}$ might be a bit aggressive in the sense that the expected squared $\ell_2$-norm of the MLMC estimator possibly does not converge, although we see no evidence of this in our experiments. A practical issue on how to choose $w_\ell$ properly depending on the problem at hand is also left open for future research.

\begin{figure}
    \centering
    \includegraphics[width=0.32\textwidth]{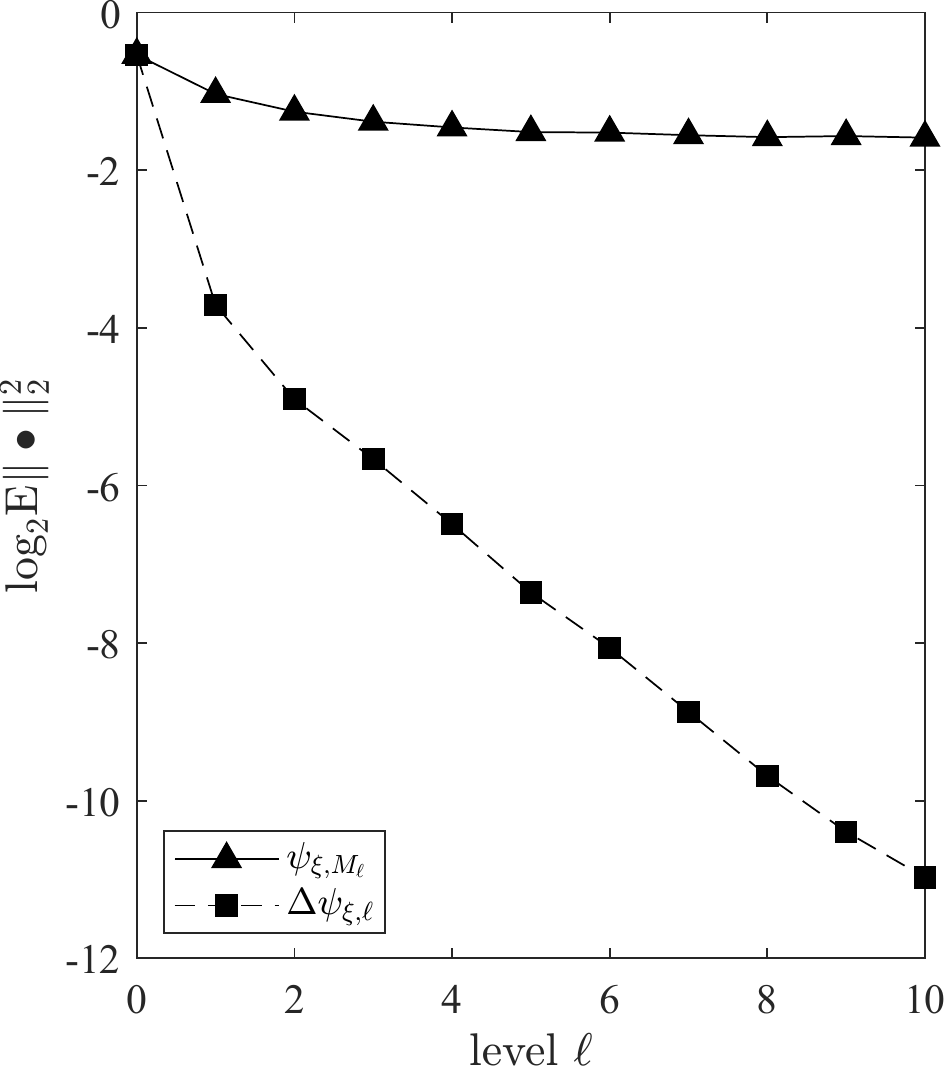}
    \includegraphics[width=0.32\textwidth]{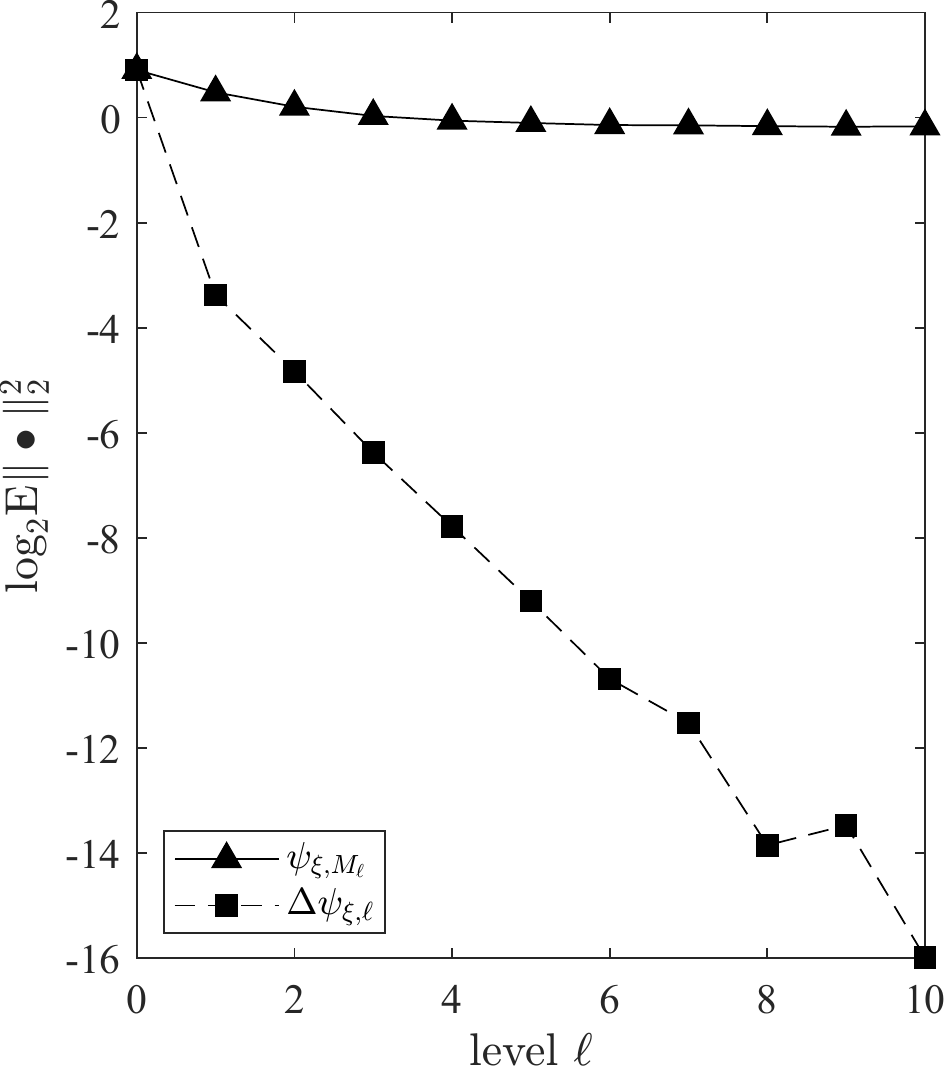}
    \includegraphics[width=0.32\textwidth]{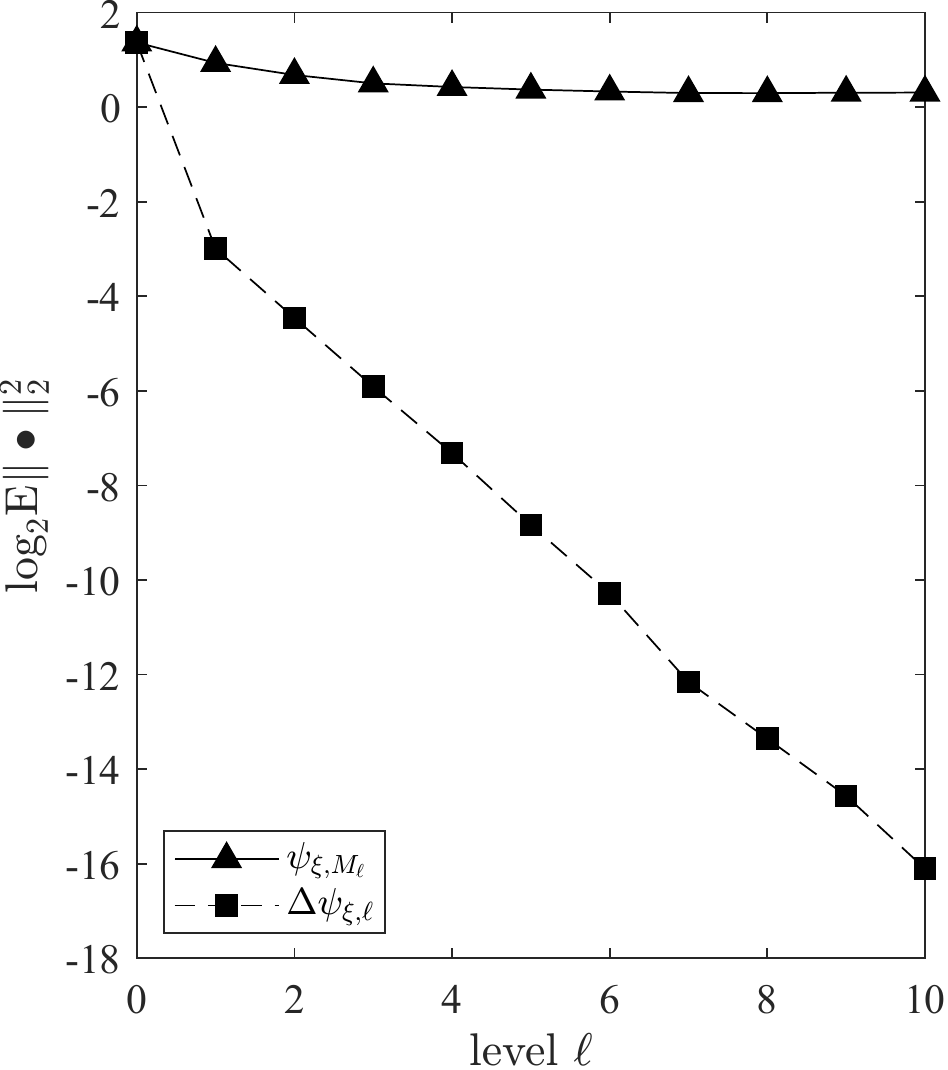}
    \caption{The mean squares of the variables $\psi_{\xi,M_\ell}$ and $\Delta \psi_{\xi,\ell}$ for the PK model at the iteration steps $t=0, T/2, T$}
    \label{fig:pk-variance}
\end{figure}

Finally, Fig.~\ref{fig:pk-eig_iterations} shows the behaviors of the expected information gain $U$ as a function of the number of iteration steps. For this problem, the expected information gain for any design parameter $\xi$ cannot be evaluated exactly, so that we use a randomized variant of the MLMC estimator introduced in \cite{GHI20} with $10^6$ outer samples to estimate the expected information gain for every 500 steps. As 10 independent runs are performed, we plot the average of 10 estimated values in mark, while the shaded area represents the linearly interpolated standard error. We can see that the expected information gain increases with some fluctuation as the iteration proceeds, and converges to a constant value. The average converged value for the MLMC estimator is $4.544$, which is slightly larger than $4.535$ obtained for the standard Monte Carlo estimator. Note that the expected information gain for the initial design is estimated as $3.774$, which is well below the maximum values obtained both for the standard Monte Carlo estimator and the MLMC estimator. Just to provide an intuition of this improvement, assume that each individual variable in $\theta$ remains independent and follows a normal distribution with an equal variance even after observing $Y$, which is usually not true. Then it is inferred that the variance of each variable after observing $Y$ with the initial design is reduced on average by the factor $(\exp(3.774/3))^2 \approx 12.379$, whereas that with the resulting design by our proposed optimization algorithm is $(\exp(4.554/3))^2 \approx 20.822$. Although the increment of the maximum expected information gain by using the MLMC estimator seems marginal as compared to the the standard Monte Carlo estimator in this example, it is important to emphasize again that the resulting experimental designs are qualitatively different.

\begin{figure}
    \centering
    \includegraphics[width=0.8\textwidth]{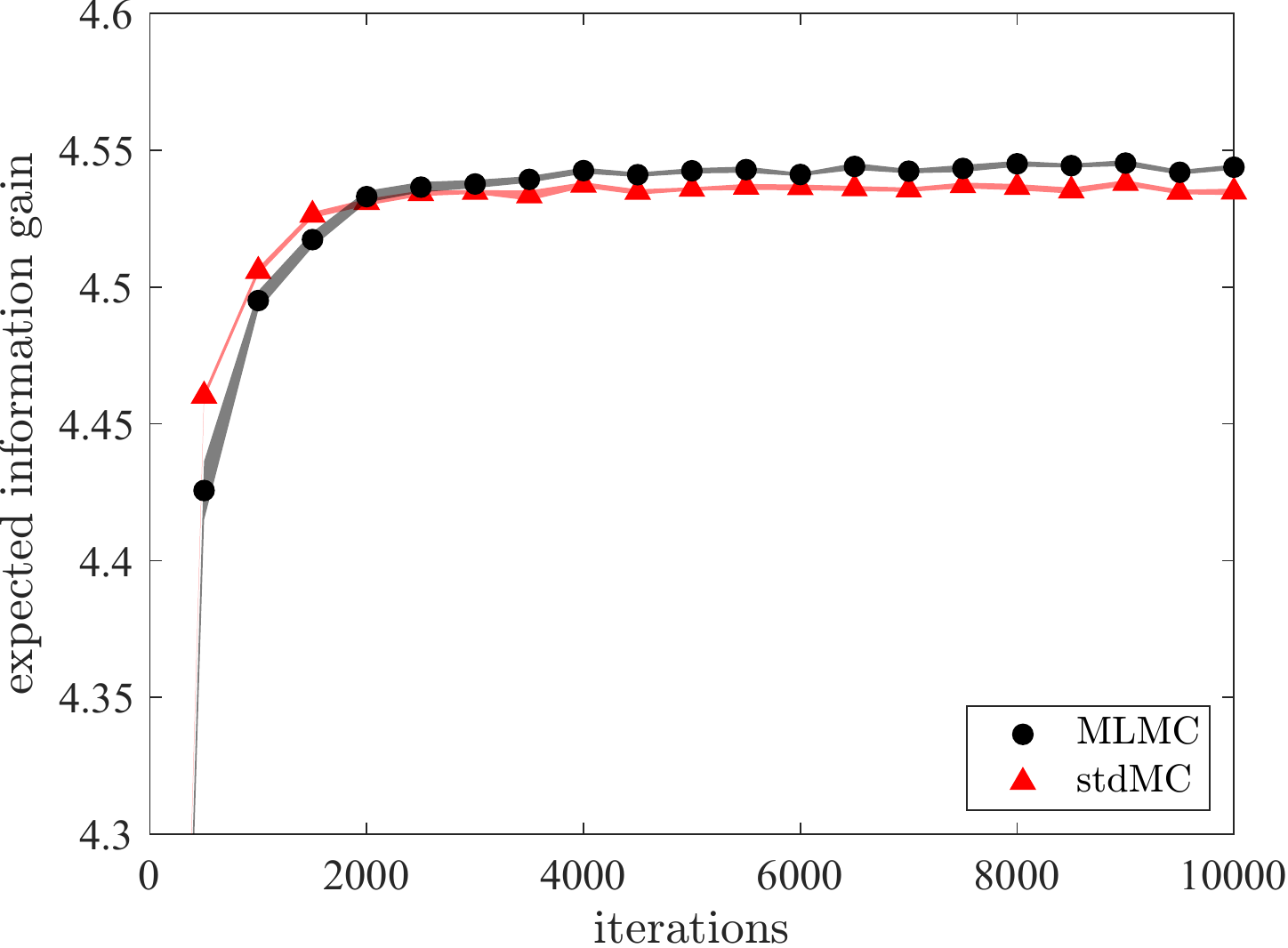}
    \caption{The behavior of the expected information gain as a function of the number of iteration steps for the PK model}
    \label{fig:pk-eig_iterations}
\end{figure}

\section{Conclusion}\label{sec:end}
In this paper we have developed an efficient stochastic algorithm to optimize Bayesian experimental designs such that the expected information gain is maximized. Since the gradient of the expected information gain with respect to design parameters is expressed as a nested expectation, a straightforward use of stochastic gradient-based optimization algorithms in which the number of inner Monte Carlo samples is kept fixed only gives a biased solution of Bayesian experimental design unless i.i.d.\ sampling from the exact posterior distribution is possible. To overcome this issue, we have introduced an unbiased antithetic multilevel Monte Carlo estimator for the gradient of the expected information gain, and have proven under some conditions that our estimator is unbiased and has finite expected squared $\ell_2$-norm and finite computational cost per one sample. This way, combining our unbiased multilevel estimator with stochastic gradient-based optimization algorithms leads to a novel stochastic algorithm to search for optimal Bayesian experimental designs without suffering from any bias. Numerical experiments for a simple test case show that our proposed algorithm can find the true optimal Bayesian experimental design with the convergence behavior as expected from the standard stochastic optimization theory which is built upon the underlying assumption that an unbiased gradient estimation is possible. In contrast, using the standard Monte Carlo estimator with a fixed number of inner samples fails to reach the optimal design. Moreover, our proposed algorithm performs well for a more realistic pharmacokinetic test problem and gives a higher expected information gain and qualitatively different sampling times compared to designs obtained by the existing standard Monte Carlo estimator.

\section*{Acknowledgements}
The authors would like to thank the reviewers for their helpful comments and suggestions which lead to a significant improvement over the original manuscript. The authors are also grateful to Takuro Mori (University of Tokyo) for his help on numerical experiments during a revision process. TG would like to thank Professor Mike Giles (University of Oxford) for useful discussions and comments at an early stage of this research.

\appendix
\section{Proof of Theorem~\ref{thm:mlmc_conv}}\label{app:proof}
The proof for the first assertion follows an argument similar to that of \cite[Lemma~3.9]{HGGT20} which considers a nested expectation involving the ratio of two scalar inner conditional expectations. Since the numerator is vector-valued in our setting, however, we give a proof for the sake of completeness.

First let us recall the following result proven, for instance, in  \cite[Lemma~1]{GG19}. 
\begin{lemma}\label{lem:concentration}
Let $X$ be a real-valued random variable with mean zero, and let $\overline{X}_N$ be an average of N i.i.d.\ samples of $X$. If $\EE[|X|^u]<\infty$ for $u> 2$, there exists a constant $C_u>0$ depending only on $u$ such that
\[ \EE\left[ \left| \overline{X}_N\right|^u\right]\leq C_u\frac{\EE[|X|^u]}{N^{u/2}}\quad \text{and}\quad \PP\left[ \left| \overline{X}_N\right|>c\right]\leq C_u\frac{\EE[|X|^u]}{c^uN^{u/2}}, \]
for any $c>0$.
\end{lemma}

For any $\theta$, $\epsilon$ and $\xi$, we write $\rho(f_{\xi}(\theta,\epsilon)\mid \xi)=\EE_{\theta'}\left[ \rho(f_{\xi}(\theta,\epsilon)\mid \theta',\xi)\right]$ and also $\nabla_{\xi}\rho(f_{\xi}(\theta,\epsilon)\mid \xi)=\EE_{\theta'}\left[\nabla_{\xi} \rho(f_{\xi}(\theta,\epsilon)\mid \theta',\xi)\right]$. For randomly chosen $\theta$ and $\epsilon$, we define an extreme event $A$ by
\[ A :=  \left\{ \left|\frac{ \varrho_{\xi,M_02^{\ell-1},q}^{(a)}(\theta,\epsilon)}{\rho(f_{\xi}(\theta,\epsilon)\mid \xi)}-1\right| >\frac{1}{2}\right\} \bigcup \left\{ \left|\frac{ \varrho_{\xi,M_02^{\ell-1},q}^{(b)}(\theta,\epsilon)}{\rho(f_{\xi}(\theta,\epsilon)\mid \xi)}-1\right| >\frac{1}{2}\right\} .\]
Then we have
\begin{align}\label{eq:proof1}
    \EE[\|\Delta \psi_{\xi,\ell}\|_2^2] = \EE[\|\Delta \psi_{\xi,\ell}\|_2^2\bsone_A]+\EE[\|\Delta \psi_{\xi,\ell}\|_2^2\bsone_{A^c}],
\end{align}
where $\bsone_{\bullet}$ denotes the indicator function of an event $\bullet$ and $A^c$ denotes the complement of the event $A$.

Let us look at the first term on the right-hand side of \eqref{eq:proof1}. Since we use the same i.i.d.\ samples of $\theta'\sim q$ in the denominator and numerator for the three terms of $\Delta \psi_{\xi,\ell}$, i.e., $\psi_{\xi,M_0 2^\ell,q}, \psi_{\xi,M_0 2^{\ell-1},q}^{(a)}, \psi_{\xi,M_0 2^{\ell-1},q}^{(b)}$, it follows from the assumption
\[ \sup_{\theta,\theta',\epsilon} \left\| \nabla_{\xi}\log \rho(f_{\xi}(\theta,\epsilon)\mid \theta',\xi)\right\|_{\infty} =: \varrho_{\max} <\infty \]
that $\left\|\psi_{\xi,M_0 2^\ell,q}\right\|_2^2, \left\|\psi_{\xi,M_0 2^{\ell-1},q}^{(a)}\right\|_2^2, \left\|\psi_{\xi,M_0 2^{\ell-1},q}^{(b)}\right\|_2^2\leq 2d\varrho^2_{\max}$ where $d$ denotes the cardinality of $\xi$. Applying Jensen's inequality leads to a bound
\begin{align*}
    \left\|\Delta \psi_{\xi,\ell}\right\|_2^2 & \leq \left( \left\|\psi_{\xi,M_0 2^\ell,q}\right\|_2+\frac{\left\|\psi_{\xi,M_0 2^{\ell-1},q}^{(a)}\right\|_2}{2}+\frac{\left\|\psi_{\xi,M_0 2^{\ell-1},q}^{(b)}\right\|_2}{2}\right)^2 \\
    & \leq 2\left\|\psi_{\xi,M_0 2^\ell,q}\right\|_2^2+\left\|\psi_{\xi,M_0 2^{\ell-1},q}^{(a)}\right\|_2^2+\left\|\psi_{\xi,M_0 2^{\ell-1},q}^{(b)}\right\|_2^2 \leq 8d\varrho^2_{\max}.
\end{align*} 
Thus we have
\begin{align*}
\EE[\|\Delta \psi_{\xi,\ell}\|_2^2\bsone_A] \leq 8d\varrho^2_{\max}\EE[\bsone_A]=8d\varrho^2_{\max}\PP[A].
\end{align*} 
Noting that both $\varrho_{\xi,M_02^{\ell-1},q}^{(a)}(\theta,\epsilon)$ and $\varrho_{\xi,M_02^{\ell-1},q}^{(b)}(\theta,\epsilon)$ are unbiased estimates of the target quantity $\rho(f_{\xi}(\theta,\epsilon)\mid \xi)$ using $M_02^{\ell-1}$ random samples of $\theta'\sim q$, it follows from the assumption of the theorem and Lemma~\ref{lem:concentration} that 
\begin{align*}
\PP[A] & \leq \PP\left[ \left|\frac{ \varrho_{\xi,M_02^{\ell-1},q}^{(a)}(\theta,\epsilon)}{\rho(f_{\xi}(\theta,\epsilon)\mid \xi)}-1\right|>\frac{1}{2}\right] + \PP\left[ \left|\frac{ \varrho_{\xi,M_02^{\ell-1},q}^{(b)}(\theta,\epsilon)}{\rho(f_{\xi}(\theta,\epsilon)\mid \xi)}-1\right|>\frac{1}{2}\right] \\
& \leq \frac{2^{u+1}C_u}{(M_02^{\ell-1})^{u/2}}\EE_{\theta,\theta',\epsilon} \left[\left| \frac{\rho(f_{\xi}(\theta,\epsilon)\mid \theta',\xi)\pi_0(\theta')}{ \rho(f_{\xi}(\theta,\epsilon)\mid \xi)q(\theta')}-1\right|^u \right] \\
& \leq \frac{2^{u+1}C_u}{(M_02^{\ell-1})^{u/2}}\left(\EE_{\theta,\theta',\epsilon} \left[\left| \frac{\rho(f_{\xi}(\theta,\epsilon)\mid \theta',\xi)\pi_0(\theta')}{ \rho(f_{\xi}(\theta,\epsilon)\mid \xi)q(\theta')}\right|^u \right]+1\right).
\end{align*}
This gives a bound on the term $\EE[\|\Delta \psi_{\xi,\ell}\|_2^2\bsone_A]$ of order $2^{-(u/2)\ell}$.

Next let us look at the second term on the right-hand side of \eqref{eq:proof1}. By using the antithetic properties \eqref{eq:antithetic}, we have
\begin{align*}
& \Delta \psi_{\xi,\ell} \\
& = \frac{1}{2}\left( \nabla \varrho_{\xi,M_02^{\ell-1},q}^{(a)}(\theta,\epsilon)-\nabla_{\xi}\rho(f_{\xi}(\theta,\epsilon)\mid \xi)\right) \left( \frac{1}{\varrho_{\xi,M_02^{\ell-1},q}^{(a)}(\theta,\epsilon)}-\frac{1}{\rho(f_{\xi}(\theta,\epsilon)\mid \xi)}\right) \\
& \quad + \frac{1}{2}\left( \nabla \varrho_{\xi,M_02^{\ell-1},q}^{(b)}(\theta,\epsilon)-\nabla_{\xi}\rho(f_{\xi}(\theta,\epsilon)\mid \xi)\right) \left( \frac{1}{\varrho_{\xi,M_02^{\ell-1},q}^{(b)}(\theta,\epsilon)}-\frac{1}{\rho(f_{\xi}(\theta,\epsilon)\mid \xi)}\right) \\
& \quad - \left( \nabla \varrho_{\xi,M_02^{\ell},q}(\theta,\epsilon)-\nabla_{\xi}\rho(f_{\xi}(\theta,\epsilon)\mid \xi)\right) \left( \frac{1}{\varrho_{\xi,M_02^{\ell},q}(\theta,\epsilon)}-\frac{1}{\rho(f_{\xi}(\theta,\epsilon)\mid \xi)}\right) \\
& \quad +\frac{1}{2}\frac{\nabla_{\xi}\rho(f_{\xi}(\theta,\epsilon)\mid \xi)}{\varrho_{\xi,M_02^{\ell-1},q}^{(a)}(\theta,\epsilon)}\left( \frac{ \varrho_{\xi,M_02^{\ell-1},q}^{(a)}(\theta,\epsilon)}{\rho(f_{\xi}(\theta,\epsilon)\mid \xi)}-1\right)^2 \\
& \quad +\frac{1}{2}\frac{\nabla_{\xi}\rho(f_{\xi}(\theta,\epsilon)\mid \xi)}{\varrho_{\xi,M_02^{\ell-1},q}^{(b)}(\theta,\epsilon)}\left( \frac{ \varrho_{\xi,M_02^{\ell-1},q}^{(b)}(\theta,\epsilon)}{\rho(f_{\xi}(\theta,\epsilon)\mid \xi)}-1\right)^2 \\
& \quad -\frac{\nabla_{\xi}\rho(f_{\xi}(\theta,\epsilon)\mid \xi)}{\varrho_{\xi,M_02^{\ell},q}(\theta,\epsilon)}\left( \frac{ \varrho_{\xi,M_02^{\ell},q}(\theta,\epsilon)}{\rho(f_{\xi}(\theta,\epsilon)\mid \xi)}-1\right)^2.
\end{align*}
Noting that 
\[ \left|\frac{ \varrho_{\xi,M_02^{\ell-1},q}^{(a)}(\theta,\epsilon)}{\rho(f_{\xi}(\theta,\epsilon)\mid \xi)}-1\right|, \left|\frac{ \varrho_{\xi,M_02^{\ell-1},q}^{(b)}(\theta,\epsilon)}{\rho(f_{\xi}(\theta,\epsilon)\mid \xi)}-1\right|\leq \frac{1}{2} \]
on $A^c$ and that $\varrho_{\xi,M_02^{\ell-1},q}^{(a)}(\theta,\epsilon), \varrho_{\xi,M_02^{\ell-1},q}^{(b)}(\theta,\epsilon)$ and $\rho(f_{\xi}(\theta,\epsilon)\mid \xi)$ are strictly positive by assumption, it holds that
\[ \frac{1}{\varrho_{\xi,M_02^{\ell-1},q}^{(a)}(\theta,\epsilon)}, \frac{1}{\varrho_{\xi,M_02^{\ell-1},q}^{(b)}(\theta,\epsilon)} \leq \frac{2}{\rho(f_{\xi}(\theta,\epsilon)\mid \xi)}. \]
The same bound exists also for $\varrho_{\xi,M_02^{\ell},q}(\theta,\epsilon)$ because of the antithetic property \eqref{eq:antithetic}. By applying Jensen's inequality and then using these bounds, we obtain
\begin{align}
\left\|\Delta \psi_{\xi,\ell}\right\|_2^2 
& \leq 2\left\|  \frac{\nabla \varrho_{\xi,M_02^{\ell-1},q}^{(a)}(\theta,\epsilon)-\nabla_{\xi}\rho(f_{\xi}(\theta,\epsilon)\mid \xi)}{\varrho_{\xi,M_02^{\ell-1},q}^{(a)}(\theta,\epsilon)}\right\|_2^2 \left( \frac{\varrho_{\xi,M_02^{\ell-1},q}^{(a)}(\theta,\epsilon)}{\rho(f_{\xi}(\theta,\epsilon)\mid \xi)}-1\right)^2 \notag \\
& \quad + 2\left\| \frac{\nabla \varrho_{\xi,M_02^{\ell-1},q}^{(b)}(\theta,\epsilon)-\nabla_{\xi}\rho(f_{\xi}(\theta,\epsilon)\mid \xi)}{\varrho_{\xi,M_02^{\ell-1},q}^{(b)}(\theta,\epsilon)}\right\|_2^2 \left( \frac{\varrho_{\xi,M_02^{\ell-1},q}^{(b)}(\theta,\epsilon)}{\rho(f_{\xi}(\theta,\epsilon)\mid \xi)}-1\right)^2 \notag \\
& \quad +4 \left\| \frac{\nabla \varrho_{\xi,M_02^{\ell},q}(\theta,\epsilon)-\nabla_{\xi}\rho(f_{\xi}(\theta,\epsilon)\mid \xi)}{\varrho_{\xi,M_02^{\ell},q}(\theta,\epsilon)}\right\|_2^2 \left( \frac{\varrho_{\xi,M_02^{\ell},q}(\theta,\epsilon)}{\rho(f_{\xi}(\theta,\epsilon)\mid \xi)}-1\right)^2 \notag \\
& \quad +2\left\|\frac{\nabla_{\xi}\rho(f_{\xi}(\theta,\epsilon)\mid \xi)}{\varrho_{\xi,M_02^{\ell-1},q}^{(a)}(\theta,\epsilon)}\right\|_2^2\left( \frac{ \varrho_{\xi,M_02^{\ell-1},q}^{(a)}(\theta,\epsilon)}{\rho(f_{\xi}(\theta,\epsilon)\mid \xi)}-1\right)^4 \notag \\
& \quad +2\left\|\frac{\nabla_{\xi}\rho(f_{\xi}(\theta,\epsilon)\mid \xi)}{\varrho_{\xi,M_02^{\ell-1},q}^{(b)}(\theta,\epsilon)}\right\|_2^2\left( \frac{ \varrho_{\xi,M_02^{\ell-1},q}^{(b)}(\theta,\epsilon)}{\rho(f_{\xi}(\theta,\epsilon)\mid \xi)}-1\right)^4 \notag \\
& \quad +4\left\|\frac{\nabla_{\xi}\rho(f_{\xi}(\theta,\epsilon)\mid \xi)}{\varrho_{\xi,M_02^{\ell},q}(\theta,\epsilon)}\right\|_2^2\left( \frac{ \varrho_{\xi,M_02^{\ell},q}(\theta,\epsilon)}{\rho(f_{\xi}(\theta,\epsilon)\mid \xi)}-1\right)^4 \notag \\
\begin{split}\label{eq:proof2}
& \leq 8\left\| \frac{\nabla \varrho_{\xi,M_02^{\ell-1},q}^{(a)}(\theta,\epsilon)-\nabla_{\xi}\rho(f_{\xi}(\theta,\epsilon)\mid \xi)}{\rho(f_{\xi}(\theta,\epsilon)\mid \xi)}\right\|_2^2 \left( \frac{\varrho_{\xi,M_02^{\ell-1},q}^{(a)}(\theta,\epsilon)}{\rho(f_{\xi}(\theta,\epsilon)\mid \xi)}-1\right)^2 \\
& \quad + 8\left\| \frac{\nabla \varrho_{\xi,M_02^{\ell-1},q}^{(b)}(\theta,\epsilon)-\nabla_{\xi}\rho(f_{\xi}(\theta,\epsilon)\mid \xi)}{\rho(f_{\xi}(\theta,\epsilon)\mid \xi)}\right\|_2^2 \left( \frac{\varrho_{\xi,M_02^{\ell-1},q}^{(b)}(\theta,\epsilon)}{\rho(f_{\xi}(\theta,\epsilon)\mid \xi)}-1\right)^2 \\
& \quad +16 \left\| \frac{\nabla \varrho_{\xi,M_02^{\ell},q}(\theta,\epsilon)-\nabla_{\xi}\rho(f_{\xi}(\theta,\epsilon)\mid \xi)}{\rho(f_{\xi}(\theta,\epsilon)\mid \xi)}\right\|_2^2 \left( \frac{\varrho_{\xi,M_02^{\ell},q}(\theta,\epsilon)}{\rho(f_{\xi}(\theta,\epsilon)\mid \xi)}-1\right)^2 \\
& \quad +8\left\|\frac{\nabla_{\xi}\rho(f_{\xi}(\theta,\epsilon)\mid \xi)}{\rho(f_{\xi}(\theta,\epsilon)\mid \xi)}\right\|_2^2\left( \frac{ \varrho_{\xi,M_02^{\ell-1},q}^{(a)}(\theta,\epsilon)}{\rho(f_{\xi}(\theta,\epsilon)\mid \xi)}-1\right)^4 \\
& \quad +8\left\|\frac{\nabla_{\xi}\rho(f_{\xi}(\theta,\epsilon)\mid \xi)}{\rho(f_{\xi}(\theta,\epsilon)\mid \xi)}\right\|_2^2\left( \frac{ \varrho_{\xi,M_02^{\ell-1},q}^{(b)}(\theta,\epsilon)}{\rho(f_{\xi}(\theta,\epsilon)\mid \xi)}-1\right)^4 \\
& \quad +16\left\|\frac{\nabla_{\xi}\rho(f_{\xi}(\theta,\epsilon)\mid \xi)}{\rho(f_{\xi}(\theta,\epsilon)\mid \xi)}\right\|_2^2\left( \frac{ \varrho_{\xi,M_02^{\ell},q}(\theta,\epsilon)}{\rho(f_{\xi}(\theta,\epsilon)\mid \xi)}-1\right)^4.
\end{split}
\end{align}

Let us focus on the third term of \eqref{eq:proof2}. Applying H\"{o}lder's inequality gives
\begin{align*}
    & \EE\left[ \left\| \frac{\nabla \varrho_{\xi,M_02^{\ell},q}(\theta,\epsilon)-\nabla_{\xi}\rho(f_{\xi}(\theta,\epsilon)\mid \xi)}{\rho(f_{\xi}(\theta,\epsilon)\mid \xi)}\right\|_2^2 \left( \frac{\varrho_{\xi,M_02^{\ell},q}(\theta,\epsilon)}{\rho(f_{\xi}(\theta,\epsilon)\mid \xi)}-1\right)^2\bsone_{A^c}\right] \\
    & \leq \EE\left[\left\| \frac{\nabla \varrho_{\xi,M_02^{\ell},q}(\theta,\epsilon)-\nabla_{\xi}\rho(f_{\xi}(\theta,\epsilon)\mid \xi)}{\rho(f_{\xi}(\theta,\epsilon)\mid \xi)}\right\|_2^2 \right. \\
    & \quad \qquad \times \left. 2^{\max(4-u,0)} \left| \frac{\varrho_{\xi,M_02^{\ell},q}(\theta,\epsilon)}{\rho(f_{\xi}(\theta,\epsilon)\mid \xi)}-1\right|^{\min(u,4)-2}\right] \\
    & \leq \left( \EE\left[\left\| \frac{\nabla \varrho_{\xi,M_02^{\ell},q}(\theta,\epsilon)-\nabla_{\xi}\rho(f_{\xi}(\theta,\epsilon)\mid \xi)}{\rho(f_{\xi}(\theta,\epsilon)\mid \xi)}\right\|_2^{\min(u,4)} \right] \right)^{2/\min(u,4)} \\
    & \quad \qquad \times 2^{\max(4-u,0)}\left( \EE\left[\left| \frac{\varrho_{\xi,M_02^{\ell},q}(\theta,\epsilon)}{\rho(f_{\xi}(\theta,\epsilon)\mid \xi)}-1\right|^{\min(u,4)} \right] \right)^{1-2/\min(u,4)}. \end{align*}
Using Jensen's inequality and Lemma~\ref{lem:concentration}, the first factor above is bounded by
\begin{align*}
    & \EE\left[\left\| \frac{\nabla \varrho_{\xi,M_02^{\ell},q}(\theta,\epsilon)-\nabla_{\xi}\rho(f_{\xi}(\theta,\epsilon)\mid \xi)}{\rho(f_{\xi}(\theta,\epsilon)\mid \xi)}\right\|_2^{\min(u,4)} \right] \\
    & \leq d^{\min(u,4)/2-1}\EE\left[\left\| \frac{\nabla \varrho_{\xi,M_02^{\ell},q}(\theta,\epsilon)-\nabla_{\xi}\rho(f_{\xi}(\theta,\epsilon)\mid \xi)}{\rho(f_{\xi}(\theta,\epsilon)\mid \xi)}\right\|_{\min(u,4)}^{\min(u,4)} \right]\\
    & \leq \frac{d^{\min(u,4)/2-1}C_{\min(u,4)}}{(M_02^{\ell})^{\min(u,4)/2}}\\
    & \quad \times \EE_{\theta,\theta',\epsilon}\left[\left\| \frac{\nabla_{\xi} \rho(f_{\xi}(\theta,\epsilon)\mid \theta',\xi)\pi_0(\theta')/q(\theta')-\nabla_{\xi}\rho(f_{\xi}(\theta,\epsilon)\mid \xi)}{\rho(f_{\xi}(\theta,\epsilon)\mid \xi)}\right\|_{\min(u,4)}^{\min(u,4)} \right] \\
    & \leq \frac{2^{\min(u,4)-1}d^{\min(u,4)/2-1}C_{\min(u,4)}}{(M_02^{\ell})^{\min(u,4)/2}} \\
    & \quad \times \EE_{\theta,\theta',\epsilon}\left[\left\| \frac{\nabla_{\xi} \rho(f_{\xi}(\theta,\epsilon)\mid \theta',\xi)\pi_0(\theta')/q(\theta')}{\rho(f_{\xi}(\theta,\epsilon)\mid \xi)}\right\|_{\min(u,4)}^{\min(u,4)}+\left\|\frac{\nabla_{\xi}\rho(f_{\xi}(\theta,\epsilon)\mid \xi)}{\rho(f_{\xi}(\theta,\epsilon)\mid \xi)}\right\|_{\min(u,4)}^{\min(u,4)} \right] \\
    & \leq \frac{2^{\min(u,4)-1}d^{\min(u,4)/2-1}C_{\min(u,4)}}{(M_02^{\ell})^{\min(u,4)/2}} \\
    & \quad \times \EE_{\theta,\theta',\epsilon}\left[\left\| \frac{\nabla_{\xi} \rho(f_{\xi}(\theta,\epsilon)\mid \theta',\xi)\pi_0(\theta')/q(\theta')}{\rho(f_{\xi}(\theta,\epsilon)\mid \theta',\xi)\pi_0(\theta')/q(\theta')}\cdot \frac{\rho(f_{\xi}(\theta,\epsilon)\mid \theta',\xi)\pi_0(\theta')/q(\theta')}{\rho(f_{\xi}(\theta,\epsilon)\mid \xi)}\right\|_{\min(u,4)}^{\min(u,4)}\right. \\
    &\quad \qquad \qquad \left. +\left\|\frac{\nabla_{\xi}\rho(f_{\xi}(\theta,\epsilon)\mid \xi)}{\rho(f_{\xi}(\theta,\epsilon)\mid \xi)}\right\|_{\min(u,4)}^{\min(u,4)} \right] \\
    & \leq \frac{2^{\min(u,4)-1}d^{\min(u,4)/2}\varrho_{\max}^{\min(u,4)}C_{\min(u,4)}}{(M_02^{\ell})^{\min(u,4)/2}} \\
    & \quad \times \left(\EE_{\theta,\theta',\epsilon} \left[\left| \frac{\rho(f_{\xi}(\theta,\epsilon)\mid \theta',\xi)\pi_0(\theta')}{ \rho(f_{\xi}(\theta,\epsilon)\mid \xi)q(\theta')}\right|^{\min(u,4)} \right]+1\right),
\end{align*}
whereas a bound on the second factor directly follows from Lemma~\ref{lem:concentration}, i.e., we have
\begin{align*}
    & \EE\left[\left| \frac{\varrho_{\xi,M_02^{\ell},q}(\theta,\epsilon)}{\rho(f_{\xi}(\theta,\epsilon)\mid \xi)}-1\right|^{\min(u,4)} \right] \\
    & \leq \frac{C_{\min(u,4)}}{(M_02^{\ell})^{\min(u,4)/2}}\EE_{\theta,\theta',\epsilon} \left[\left| \frac{\rho(f_{\xi}(\theta,\epsilon)\mid \theta',\xi)\pi_0(\theta')}{ \rho(f_{\xi}(\theta,\epsilon)\mid \xi)q(\theta')}-1\right|^{\min(u,4)} \right] \\
    & \leq \frac{C_{\min(u,4)}}{(M_02^{\ell})^{\min(u,4)/2}}\left(\EE_{\theta,\theta',\epsilon} \left[\left| \frac{\rho(f_{\xi}(\theta,\epsilon)\mid \theta',\xi)\pi_0(\theta')}{ \rho(f_{\xi}(\theta,\epsilon)\mid \xi)q(\theta')}\right|^{\min(u,4)} \right]+1\right).
\end{align*}
Substituting these bounds shows that the third term is of order 
\[ \left(2^{-\min(u,4)\ell/2}\right)^{2/\min(u,4)}\cdot \left(2^{-\min(u,4)\ell/2}\right)^{1-2/\min(u,4)}=2^{-\min(u,4)\ell/2} \]
for given $u>2$. 

Similarly, the expectation of the sixth term of \eqref{eq:proof2} can be bounded above by
\begin{align*}
    & \EE\left[ \left\|\frac{\nabla_{\xi}\rho(f_{\xi}(\theta,\epsilon)\mid \xi)}{\rho(f_{\xi}(\theta,\epsilon)\mid \xi)}\right\|_2^2\left( \frac{ \varrho_{\xi,M_02^{\ell},q}(\theta,\epsilon)}{\rho(f_{\xi}(\theta,\epsilon)\mid \xi)}-1\right)^4\bsone_{A^c}\right] \\
    & \leq 2^{\max(4-u,0)}\EE\left[ \left\|\frac{\nabla_{\xi}\rho(f_{\xi}(\theta,\epsilon)\mid \xi)}{\rho(f_{\xi}(\theta,\epsilon)\mid \xi)}\right\|_2^2\left| \frac{ \varrho_{\xi,M_02^{\ell},q}(\theta,\epsilon)}{\rho(f_{\xi}(\theta,\epsilon)\mid \xi)}-1\right|^{\min(u,4)}\right] \\
    & \leq 2^{\max(4-u,0)}d\varrho^2_{\max}\EE\left[ \left| \frac{ \varrho_{\xi,M_02^{\ell},q}(\theta,\epsilon)}{\rho(f_{\xi}(\theta,\epsilon)\mid \xi)}-1\right|^{\min(u,4)}\right] \\
    & \leq \frac{2^{\max(4-u,0)}d\varrho^2_{\max}C_{\min(u,4)}}{(M_02^{\ell})^{\min(u,4)/2}}\left(\EE_{\theta,\theta',\epsilon} \left[\left| \frac{\rho(f_{\xi}(\theta,\epsilon)\mid \theta',\xi)\pi_0(\theta')}{ \rho(f_{\xi}(\theta,\epsilon)\mid \xi)q(\theta')}\right|^{\min(u,4)} \right]+1\right).
\end{align*}
It is obvious that the other terms of \eqref{eq:proof2} can be bounded similarly. This way we obtain a bound on the term $\EE[\|\Delta \psi_{\xi,\ell}\|_2^2\bsone_{A^c}]$ of order $2^{-\min(u,4)\ell/2}$, which completes the proof of the first assertion of the theorem.

Let us move on to the second assertion. By choosing $w_{\ell} \propto 2^{-\tau \ell}$, it follows from the first assertion that
\begin{align*}
\sum_{\ell=0}^{\infty}\frac{\EE[\|\Delta \psi_{\xi,\ell}\|_2^2]}{w_\ell} \propto \sum_{\ell=0}^{\infty}2^{-(\beta-\tau) \ell},
\end{align*}
and
\begin{align*}
\sum_{\ell=0}^{\infty}C_\ell w_\ell \propto \sum_{\ell=0}^{\infty}2^{-(\tau-1) \ell}.
\end{align*}
Thus, if $1<\tau<\beta$, these two quantities are obviously bounded. It is important to remark that we have these finite bounds on the expected squared $\ell_2$-norm and the expected computational cost of the random variable $\Delta \psi_{\xi,\ell}/w_\ell$, since we assume $u> 2$, which ensures $\beta> 1$.

\bibliographystyle{siamplain}
\bibliography{references}

\end{document}

%% file: EIG-Grad_shared.tex
\usepackage{lipsum}
\usepackage{amsfonts}
\usepackage{graphicx}
\usepackage{epstopdf}
\usepackage{algorithmic}
\ifpdf
  \DeclareGraphicsExtensions{.eps,.pdf,.png,.jpg}
\else
  \DeclareGraphicsExtensions{.eps}
\fi

\allowdisplaybreaks

\newcommand{\rd}{\, \mathrm{d}}
\newcommand{\bsone}{\boldsymbol{1}}
\newcommand{\EE}{\mathbb{E}}
\newcommand{\PP}{\mathbb{P}}

\newcommand{\RR}{\mathbb{R}}
\newcommand{\ZZ}{\mathbb{Z}}
\newcommand{\Ecal}{\mathcal{E}}
\newcommand{\Xcal}{\mathcal{X}}
\newcommand{\Ycal}{\mathcal{Y}}

\newcommand{\Tcal}{\mathcal{T}}

\newsiamremark{example}{Example}
\newsiamremark{remark}{Remark}
\newsiamremark{hypothesis}{Hypothesis}
\crefname{hypothesis}{Hypothesis}{Hypotheses}
\newsiamthm{claim}{Claim}

\headers{MLMC stochastic optimization of Bayesian design}{T. Goda, T. Hironaka, W. Kitade, and A. Foster}

\title{Unbiased MLMC stochastic gradient-based optimization of Bayesian experimental designs\thanks{Submitted to the editors DATE.
\funding{The work of T.G.\ is supported by JSPS KAKENHI Grant Number 20K0374. The work of A.F. is kindly supported by EPSRC grant no. EP/N509711/1.}}}

\author{Takashi Goda\thanks{School of Engineering, University of Tokyo, Tokyo, Japan 
  (\email{goda@frcer.t.u-tokyo.ac.jp},
  \email{hironaka-tomohiko@g.ecc.u-tokyo.ac.jp},
  \email{kitade-wataru114@g.ecc.u-tokyo.ac.jp}).}
\and Tomohiko Hironaka\footnotemark[2]
\and Wataru Kitade\footnotemark[2]
\and Adam Foster\thanks{Department of Statistics, University of Oxford, Oxford, UK (\email{adam.foster@stats.ox.ac.uk}).}}

\usepackage{amsopn}
\DeclareMathOperator{\diag}{diag}

\makeatletter
\newcommand*{\addFileDependency}[1]{
  \typeout{(#1)}
  \@addtofilelist{#1}
  \IfFileExists{#1}{}{\typeout{No file #1.}}
}
\makeatother
